\documentclass[12pt,draftcls,onecolumn]{IEEEtran}
\usepackage{amsmath,amsfonts,amssymb,amsbsy,bm,paralist,theorem,color}
\usepackage{graphicx,algorithmic,algorithm,relsize}
\usepackage{multicol}
\usepackage{multirow}
\usepackage{subfigure}
\usepackage{stfloats}
\usepackage{cite}
\usepackage[square, comma, sort&compress, numbers]{natbib}
\newtheorem{lem}{Lemma}
\newtheorem{rem}{Remark}
\newtheorem{thm}{Theorem}

\newtheorem{prop}{Proposition}

\newcommand{\Pout}{\mathbb{P}}

\newcommand{\SNR}{{\rm{SNR}}}
\newcommand{\SINR}{{\rm{SINR}}}
\DeclareMathOperator*{\st}{s.t.}

\graphicspath{{fig/}}

\definecolor{orange}{RGB}{255,107,0}
\definecolor{green}{RGB}{0,160,20}

\newcommand{\kset} {\ensuremath{\mathcal{K}}}
\newcommand{\CN} {\ensuremath{\mathcal{CN}}}

\begin{document}
\title{HARQ-CC Enabled NOMA Designs With Outage Probability Constraints}
	\date\today
	\author{Yanqing Xu, Donghong Cai, Fang Fang, Zhiguo Ding, Chao Shen, and Gang Zhu\\
		\thanks{
			Part of this work has been accepted by IEEE Global Communications Conference (GlobeCom), 2018, \cite{xu-globecom-2018}.
			
			Yanqing Xu, Chao Shen and Gang Zhu are with the State Key Laboratory of Rail Traffic Control and Safety, Beijing Jiaotong University, Beijing, China. (e-mail: \{xuyanqing, chaoshen, gzhu\}@bjtu.edu.cn).
			Donghong Cai is with the Institute of Mobile Communications, Southwest Jiaotong University, Chengdu, China. (e-mail: cdhswjtu@163.com).
		   Fang Fang and Zhiguo Ding are with the School of Electrical and Electronic Engineering, The University of Manchester, UK. (e-mail: \{fang.fang, zhiguo.ding\}@manchester.ac.uk).
		
		   }
}
	\maketitle
%
	
\begin{abstract}
In this paper, we aim to design an adaptive power allocation scheme to minimize the average transmit power of a hybrid automatic repeat request with chase combining (HARQ-CC) enabled non-orthogonal multiple access (NOMA) system under strict outage constraints of users. Specifically, we assume the base station only knows the statistical channel state information of the users.
We first focus on the two-user cases. To evaluate the performance of the two-user HARQ-CC enabled NOMA systems, we first analyze the outage probability of each user. Then, an average power minimization problem is formulated. However, the attained expressions of the outage probabilities are nonconvex, and thus make the problem difficult to solve. Thus, we first conservatively approximate it by a tractable one and then use a successive convex approximation based algorithm to handle the relaxed problem iteratively. For more practical applications, we also investigate the HARQ-CC enabled transmissions in multi-user scenarios. The user-paring and power allocation problem is considered.  With the aid of matching theory, a low complexity algorithm is presented to first handle the user-paring problem. Then the power allocation problem is solved by the proposed SCA-based algorithm. Simulation results show the efficiency of the proposed transmission strategy and the near-optimality of the proposed algorithms.
\end{abstract}

\section{Introduction}

In wireless communication systems, due to the randomness of the wireless channels, deep fading may occur from time to time, which would bring deteriorating effects on the communication reliability.
To combat this, the hybrid automatic repeat request (HARQ), a powerful time diversity technique, has been incorporated into practical wireless communication systems, such as WiMAX and LTE  \cite{3GPP-HARQ}, to improve the communication reliability.
Moreover, the HARQ technique is also viewed as an important tool to guarantee the communication reliability of the mission critical service in 5G and its beyond \cite{3GPP-URLLC2}.
Generally, the HARQ techniques can be implemented by two basic schemes, i.e., the HARQ with chase combining (HARQ-CC) and HARQ with incremental redundancy (HARQ-IR).
In the HARQ-CC scheme \cite{Chase1973}, the packets delivered in different retransmission rounds are exactly the same and the receiver uses the maximum-ratio combining (MRC) to decode the packet; while in the HARQ-IR scheme \cite{IR1974}, additional parity bits are sent to the receiver during each retransmission round and the receiver uses code combining to decode the packet.
These two HARQ schemes have been widely used in the system design to against the channel uncertainty, e.g.,  \cite{su-tcom-2011,larsson-tcomm-2013,Chaitanya2011,Dechene2014,Szczecinski-2013}.
In particular, \cite{su-tcom-2011} considered the power allocation problem with the assumption that only the statistical channel state information (CSI) of the user is available at the base station (BS), and the channel follows the Rayleigh fading and keeps invariant during retransmission rounds. Then \cite{larsson-tcomm-2013} extended the work of \cite{su-tcom-2011} to the case that the channels in different retransmission rounds are with independent and identically distribution. However,
due to the challenge to analysis the performance of HARQ enabled transmissions, most existing works focus only on the point-to-point systems.

Recently,  non-orthogonal multiple access (NOMA) has been intensively studied for its high spectral efficiency \cite{ding-jsac-2017,XU-TSP-2017,Ding2014}. Compared to its counterpart, orthogonal multiple access (OMA), NOMA has two key characteristics. The first is that NOMA allows the transmitter to serve multiple users at the same resource block
simultaneously; the second is that the successive interference cancellation (SIC) technique can be exploited by the ``strong user'' to remove the interference from the ``weak user''. Consequently, the system spectral efficiency can be increased \cite{tse-2005}.
To investigate the impact of imperfect CSI on the performance of NOMA systems, by assuming that the BS only knows the statistical CSI, \cite{cui-spl-2016} investigated the power allocation problem to maximize the data rate of the system under outage constraints. To benefit the advantages of HARQ and NOMA, the HARQ technique was recently applied to the downlink NOMA systems in \cite{xu-globecom-2018-2,li2015,choi-tcom-2016,cai2018,cai2018-harq,cai-globecom-2018}.
 Specifically, a partial HARQ-CC scheme was proposed in \cite{xu-globecom-2018-2} to solve the problems in the cases that the users have different quality of service (QoS) requirements. \cite{li2015} investigated the impact of transmit power allocation on the performance of a two-user NOMA system with opportunistic HARQ. \cite{choi-tcom-2016} considered the rate selection problem under the HARQ-IR strategy to minimize the retransmission rounds. \cite{cai2018,cai2018-harq,cai-globecom-2018} analyzed the outage performance of the HARQ-CC enabled NOMA systems. However, the former works didn't treat the outage performance of users well and the multi-user cases are not considered.

In this paper, we consider the adaptive power allocation problem in a HARQ-CC aided NOMA system to minimize the average transmit power under outage constraints of users.
Here the HARQ-CC strategy is adopted in our system mainly due to its low complexity.
Firstly, we focus on the two-user case.
Different from \cite{choi-tcom-2016}, where the outage probabilities of the two users were approximated based on the large deviation theory, we tightly approximate the signal-to-interference-plus-noise-ratio (SINR) outage probabilities of users, and the tightness can be guaranteed even in the case with a few retransmission rounds.
This enables our results applicable to the use cases of 5G and its beyond that are specified by low latency requirements \cite{Andrews-2014,JSAC2017-XU}.
Then adaptive power allocation is investigated to minimize the total average transmit power under the constraints of the outage probability. However, the derived outage probabilities bring great challenges to the optimal system design.
Then by leveraging the successive convex approximation method, the problem is well treated.
In view of the multi-user scenario of most practical applications, we further study the joint user paring and power allocation problem, which is a mixed integer programming problem and thus more challenging than the one of the two-user case. The contributions of this paper are summarized as follows.
\begin{itemize}
	\item By defining the SINR outage probability of the HARQ-CC enabled NOMA system, we formulate the total average transmit power minimization problem of the two-user case. Then by carefully deriving the outage probability, we reformulated optimization problem into an explicit one. Moreover, the derived outage probability show good tightness even in the case with low signal-to noise-ratio (SNR) and small number of retransmissions. Also, based on the derived outage probabilities, the asymptotic performance of SINR outage probability is investigated.
	Then by properly approximating the outage probabilities, we propose a successive approximation based algorithm to iteratively solve the problem. And the proposed algorithm can be guaranteed to converge to at least a stationary point of the problem.
	Simulation results show that the proposed approximation approach and successive approximation based algorithm can achieve the near-optimality performance. The proposed adaptive power allocation scheme can significantly outperform the equal power allocation scheme.
	
	\item The average power minimization problem of joint user paring and power allocation in the multi-user cases is formulated. With the aid of matching theory, a low computational complexity algorithm is proposed to solve the user-paring problem first. The matching theory based algorithm is proved to converge in several iterations and the complexity is drastically decreased compared with the optimal exhaustive search method. Then we show that the power allocation problem can be well solved by the successive approximation based algorithm proposed for the two-user case.
	Simulation results show that the matching algorithm can achieve good performance compared with the optimal one by the exhaustive search.
\end{itemize}

{\bf Synopsis:} Section \ref{sec:system_model} presents the adaptive power allocation design in a two-user case. In Section \ref{sec:outage_analysis}, the outage probabilities of users are analyzed and the asymptotic performance is investigated; a successive approximation based algorithm is proposed in Section \ref{sec:power_min} to solve the adaptive power allocation problems. Section \ref{sec:multi-user} considers the joint user paring and power allocation problems in a multi-user case. The simulation results and conclusions are given by Section \ref{sec:simulation} and Section \ref{sec:conclusion}, respectively.

\section{Adaptive Power Allocation for Two-User Case} \label{sec:system_model}
 \subsection{System Model and Problem Formulation}
We first consider a downlink HARQ-CC-NOMA system in a two-user case. Then we will study the multi-user case in the next section and show how the insights obtained in the two-user case can be applied to solve the problems of multi-user case. In the considered two-user system,
a single-antenna base station (BS) serves two single-antenna users, as shown in Fig. \ref{fig:sys_model1}.
Here, we assume that the BS only knows the statistical channel knowledge, i.e., the distances of users to the BS and the distribution of small-scale fading. This assumption is practical for mission critical service or other real-time applications as less CSI feedback is needed.
The users are ordered based on their average channel gain. Without loss of generality, we assume that user 1 is a ``weak user'' that is located far from the BS and  user 2 is a ``strong user'' which is close to the BS.
Due to the limited CSI at the BS, the communication reliability will be degraded. Therefore, the HARQ-CC is applied in our system, and  the number of transmission rounds is limited to $T$ to guarantee the latency requirement.

\begin{figure}[!t]
	\centering
	\includegraphics[width=0.5\linewidth]{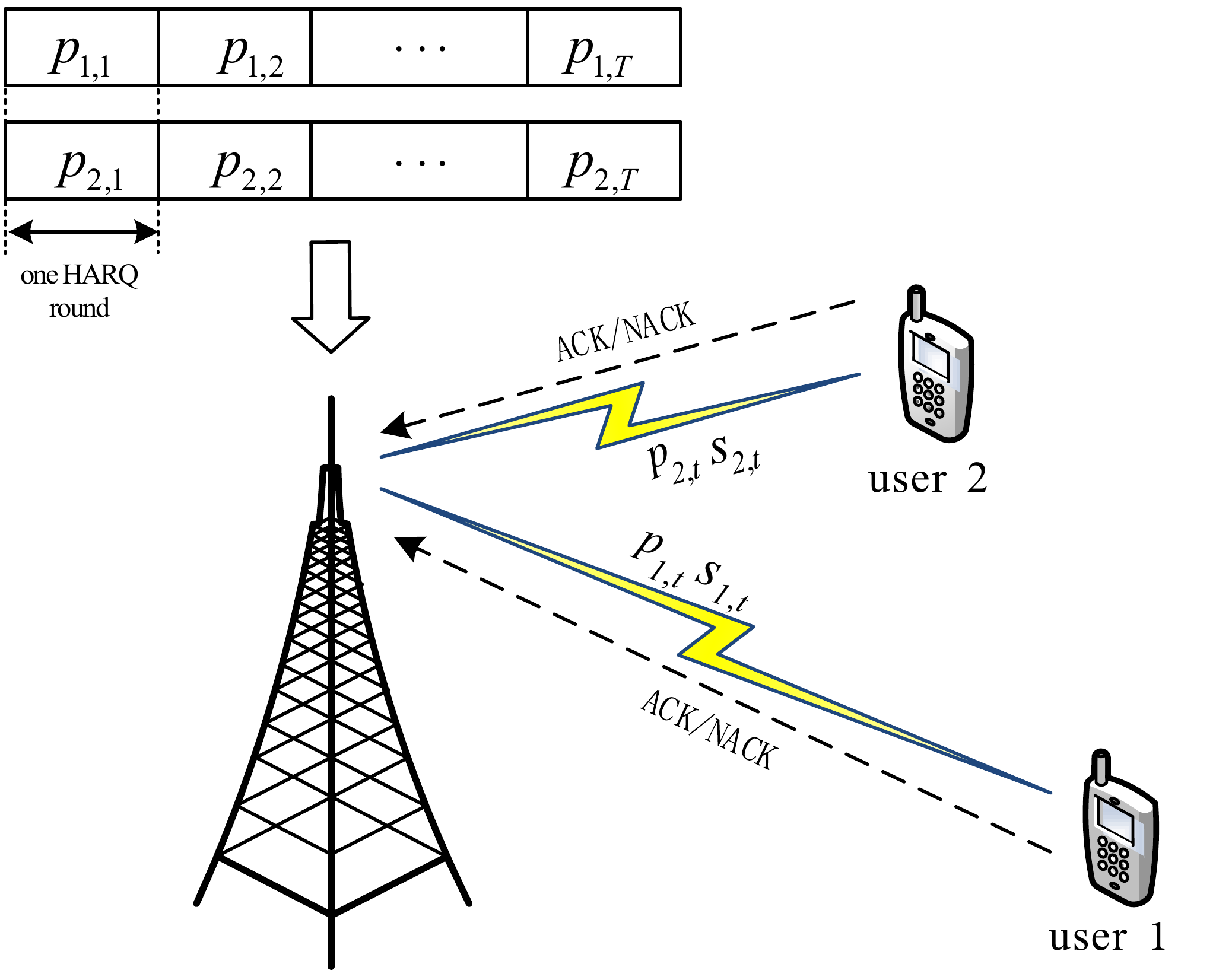}\\
\caption{An illustration of a two-user HARQ-CC enabled NOMA system.} \label{fig:sys_model1} 
\end{figure}

Under the HARQ-CC protocol, the BS keeps transmit the same packet, and the users combine the received signals during different ARQ rounds with MRC.
If a user can successfully decode its information, then an acknowledgement  (ACK) signal will be sent to the BS, where  we assume the feedback channel is error-free. Otherwise, the user will return a negative acknowledgement (NACK) signal. If the BS receives a NACK, then retransmission should be carried out until the maximum number of retransmissions is reached.
As a result, the packets will be successfully delivered or dropped after $T$ transmission rounds.
Further, we define that the transmission is successful if and only if both users return ACKs.

At the transmitter, the packets of user $k$ are independently encoded with target SNR $\gamma_k$, for $k = 1,2$. At the $t$-th HARQ-CC round, the BS combines the information of the two users by using the superposition coding strategy as $s_{t} = \sqrt{p_{1,t}} s_{1,t} + \sqrt{p_{2,t}} s_{2,t}$ where $s_{k,t}$ is the unit-power symbol and  $p_{k,t}$ is the transmit power, for $k = 1,2$. The received signal at user $k$ during the $t$-th HARQ-CC round is given by
\begin{align}
y_{k,t} = \frac{\tilde{h}_{k,t}}{\sqrt{1+d_k^{\alpha}}} (\sqrt{p_{1,t}} s_{1,t}+ \sqrt{p_{2,t}} s_{2,t}) + n_{k,t},  ~ k = 1,2,
\end{align}
where $\tilde{h}_{k,t} \sim \CN(0,1)$ models the small-scale Rayleigh fading of user $k$ during the $t$-th HARQ-CC round where $\CN(\cdot,\cdot)$ denote the complex Gaussian distribution,	
$d_k$ is the distance of user $k$ to the BS, $\alpha$ is the path loss component,
and $n_{k,t} \sim \mathcal{CN}(0,\sigma_{k}^2)$ is the additive white Gaussian noise (AWGN) at user $k$. 

According to the NOMA principle, user 2 performs SIC to first remove the interference from user 1 and then decodes its own information.
So the SINR and SNR for decoding $x_1$ and $x_2$ at the $t$-th round at user 2 are respectively given by
\begin{subequations}
	\begin{align}
	\SINR_{2,t}^{x_1} &= \frac{p_{1,t}h_{2,t}\lambda_2}{p_{2,t}h_{2,t}\lambda_2 + 1}, \\
	\SNR_{2,t}^{x_2}  &= p_{2,t} h_{2,t}\lambda_2,
	\end{align}
\end{subequations}
where $h_{2,t} \lambda_2$  denotes the normalized channel power gain from the BS to user 2 with $h_{2,t} = |\tilde{h}_{2,t}|^2$ and $\lambda_2 = \frac{1}{(1+d_2^{\alpha})\sigma_2^2}$. Meanwhile, for user 1, it decodes its own information directly by treating the signal for user 2 as noise, thus the SINR to decode $x_1$ at the $t$-th round at user 1 can be described as
\begin{align}
\SINR_{1,t}^{x_1} = \frac{p_{1,t}h_{1,t}\lambda_1}{p_{2,t}h_{1,t}\lambda_1 + 1},
\end{align}
where $h_{1,t} = |\tilde{h}_{1,t}|^2$ and $\lambda_1 = \frac{1}{(1+d_1^{\alpha})\sigma_1^2}$.
Here we define the outage event after $t$ transmission rounds as that the combined SINR is smaller than the target SNR. Thus, up to the $t$-th HARQ-CC round, the users' outage probabilities can be respectively defined by
\begin{subequations} \label{eq:outage_users}
	\begin{align}
	\Pout_{1,t} &= \Pr\left(\sum_{\ell=1}^{t}\SINR_{1,\ell}^{x_1} < \gamma_1\right),   \label{eq:outage_1}\\
	\Pout_{2,t} &= 1 \!-\! \Pr\left(\sum_{\ell=1}^{t}\SINR_{2,\ell}^{x_1} \ge \gamma_1, \sum_{\ell=1}^{t}\SNR_{2,\ell}^{x_2} \!\ge\! \gamma_2\right).   \label{eq:outage_2}
	\end{align}
\end{subequations}
To guarantee the communication reliability requirements of both users, the outage probabilities after $T$ transmission rounds should satisfy
\begin{align}
\Pout_{k,T} \le \delta_k, ~k = 1,2,
\end{align}
where $\delta_k$ is the maximum tolerable outage probability for user $k$.
Recall that the transmission is successful if and only if both users can successfully decode their packets. Thus the retransmission probability at the $t$-th HARQ-CC round is determined by the outage probabilities of the two users in the $(t-1)$-th HARQ-CC round. Specifically, the retransmission probability at the $t$-th HARQ-CC round is given by
\begin{subequations}
	\begin{align}\label{eq:retransmission_probability}
	{\hat{\Pout}_t} &= 1 - (1 - \Pout_{1,t-1}) (1 - \Pout_{2,t-1}) \\
	&= \Pout_{1,t-1} + \Pout_{2,t-1} - \Pout_{1,t-1} \Pout_{2,t-1}, ~\forall t = 1,...,T,
	\end{align}
\end{subequations}
where $\Pout_{1,0} = \Pout_{2,0} = 1$.

Define the average transmit power at $t$-th HARQ-CC round as $P_{{\rm avg},t} \triangleq (p_{1,t} + p_{2,t})\hat{\Pout}_t $.
Thus the total average transmit power of the system after $T$ transmissions can be written as
\begin{subequations}\label{eq:averge_power_2}
	\begin{align}
	{{P_{\rm avg} }}~& = \sum_{t = 1}^{T} P_{{\rm avg},t} = \sum_{t = 1}^{T} (p_{1,t}+p_{2,t})\hat{\Pout}_t \\
	&= p_{1,1} + p_{2,1} + \sum_{t = 2}^{T} (p_{1,t}+p_{2,t})\hat{\Pout}_t.
	\end{align}
\end{subequations}

We aim to minimize the total average power while guaranteeing the QoS requirements of the two users. Therefore, the outage-constrained HARQ-CC transmission design can be formulated as
\begin{subequations}\label{p:power_min_ARQ}
	\begin{align}
	\min_{{\bf p_1}, {\bf p_2}} \quad &  {{p_{1,1} + p_{1,2} + \sum_{t = 2}^{T} (p_{1,t}+p_{2,t})\hat{\Pout}_{t}}} \\
	\st ~~~& \Pout_{k, T}\le \delta_k,~k =1,2,\label{eq:outage_constraint}\\
	& p_{1,t} \ge 0,~ p_{2,t} \ge 0,  \forall t= 1,...,T,\label{eq:power_allocation_1}\\
	& p_{1,t} + p_{2,t} \le P_{\max}, \forall t=1,...,T,\label{eq:power_allocation_2}
	\end{align}
\end{subequations}
where ${\bf p_1} = [p_{1,1},...,p_{1,T}]^T$, ${\bf p_2} = [p_{2,1},...,p_{2,T}]^T$, and $P_{\max}$ is the maximum transmit power. \eqref{eq:outage_constraint} are the outage constraints.
Note that the main challenge in solving problem \eqref{p:power_min_ARQ} is due to the fact that the outage probabilities of users do not have closed-form expressions. Thus, in the next subsection, we first to explicitly characterize the two outage probabilities, $\Pout_{k,T}$, for $k = 1,2$.

\subsection{Outage Probability Analysis} \label{sec:outage_analysis}
Let's first focus on the outage probability of user 1 given in \eqref{eq:outage_1}.
\subsubsection{Outage Probability of User 1}
The outage probability of user 1 is given in the following theorem.
\begin{thm}
	After $T$ transmissions, the outage probability of user 1 is given by.
	\begin{align} \label{eq:ana_outage_user1}
	\Pout_{1,T} &\approx 2^{T+1}\ln 2\left(\frac{\pi}{N}\right)^{T+1}\left(\prod_{t=1}^{T}\beta_tp_{1,t}\right)e^{\lambda_1^{-2}\sum\limits_{t = 1}^{T}p_{2,t}^{-1}}\Bigg[\sum_{n\in\mathcal{N}}\left(\prod_{t=1}^{T}\mathcal{C}_t(a_n)\right)\nonumber \\
	& ~~~~~~~~~~~~~~~~~~~~~~~~~~~~~~~~~~~~~~~~~~~~~~\times \sum_{m=1}^{M}\sum_{k=1}^{N} \left({\frac{w_m\sqrt{1-a_{k}^2}}{\gamma_1(1+a_k)} } 2^{-\frac{m(1+a_n)}{\gamma_1(1+a_{k})}\sum_{t=1}^{T}\beta_t}\right)\Bigg],
	\end{align}
	where $a_{n}=\cos\left(\frac{2n-1}{2N}\pi\right)$ and  $a_{k}=\cos\left(\frac{2k-1}{2N}\pi\right)$ for $n, k = 1,...,N$ and $N$ (e.g. $N=30$) is a parameter for Gaussian-Chebyshev quadrature, $\beta_t = \frac{p_{1,t}}{p_{2,t}}$, $\mathcal{C}_t(x)$ is defined as
	\begin{align}
	\mathcal{C}_t(x)=\frac{\sqrt{1-x^2}}{\lambda_1(2p_{1,t}-p_{1,t}(x+1))^2}e^{-\frac{2p_{1,t}}{p_{2,t}(2p_{1,t}-p_{1,t}(x+1))\lambda_1}},
	\end{align}
	 and $w_m$ is given by
	\begin{align}
		w_m = (-1)^{\frac{M}{2} + m}  \sum_{n = \left\lfloor \frac{(m+1)}{2} \right\rfloor}^{\min\{m,\frac{M}{2}\}}\frac{  n^{\frac{M}{2}}(2n)!}{m! \left(\frac{M}{2}-n\right)!n!(n-1)!(m-n)!(2n-m)!}.
	\end{align}
\end{thm}

\begin{IEEEproof}
The cumulative distribution function (CDF) of $\SINR_{1,t}^{x_1}$ is given by
\begin{align}
F_{\SINR_{1,t}^{x_1}}(z_t)&=\Pr\left(\frac{p_{1,t}h_{1,t}\lambda_1}{p_{2,t}h_{1,t}\lambda_1 + 1} < z_t\right)\nonumber\\
&=\Pr\left(\lambda_1 (p_{1,t}-z_tp_{2,t})h_{1,t}< z_t\right). \label{eq:outage_user1}
\end{align}
Note that  $F_{\SINR_{1,t}^{x_1}}(z_t)=1$ for $z_t\in[\beta_t,+\infty)$. While for $z_t\in(0,\beta_t)$, with the fact that $h_{1,t}$ follows the exponential distribution, we have
\begin{align}\label{eq:joint_PDF}
F_{\SINR_{1,t}^{x_1}}(z_t)=
\begin{cases}
1-e^{-\frac{z_t}{\left(p_{1,t}-z_tp_{2,t}\right)\lambda_1}}, &{\rm if}~ z_t \in [0, \beta_t)\\
1, &{\rm otherwise} .
\end{cases}
\end{align}
Then the probability density function (PDF) of $\SINR_{1,t}^{x_1}$ can be described as
\begin{align}
f_{\SINR_{1,t}^{x_1}}(z_t)=
\begin{cases}
\frac{p_{1,t}}{\lambda_1(p_{1,t}-z_tp_{2,t})^2}e^{-\frac{z_t}{\lambda_1(p_{1,t}-z_tp_{2,t})}}, &{\rm if}~ z_t \in [0, \beta_t),\\
0, &{\rm otherwise}.
\end{cases}
\end{align}
Let $f_Z (z)$ denote the PDF of $Z \triangleq\sum_{t=1}^{T}\SINR_{1,t}^{x_1}$, and note that $\SINR_{1,t}^{x_1}$ for $t=1,\ldots, T$, are statistically independent. Hence the Laplace transform of $f_Z (z)$ can be shown as $\widehat{f}_Z(s) = \prod_{t=1}^{T}\widehat{f}_{\SINR_{1,t}^{x_1}}(s)$ with
\begin{subequations}
	\begin{align}
	\widehat{f}_{\SINR_{1,t}^{x_1}}(s)&=\int_0^{\beta_t}f_{\SINR_{1,t}^{x_1}}(z_t)e^{-sz_t}dz_t, \\
	&\approx c_t\sum_{n=1}^N\mathcal{C}_t(a_n)e^{-\frac{s\beta_t(a_n+1)}{2}}, \label{eq:laplace}
	\end{align}
\end{subequations}
where the second step is obtained by applying the Gaussian-Chebyshev quadrature \cite{Ding2014}
and $c_t=\frac{2\pi}{N}\beta_t p_{1,t}e^{\frac{1}{\lambda_1p_{2,t}}}$.

Submitting \eqref{eq:laplace} into $\widehat{f}_Z(s)$, we have
\begin{align}
\widehat{f}_Z(s)
&\approx\left(\frac{2\pi}{N}\right)^{T}\left(\prod_{t=1}^{T}\beta_t p_{1,t}\right)e^{\sum\limits_{t=1}^{T}\frac{1}{\lambda_1p_{2,t}}}
\prod_{t=1}^{T}\left(\sum_{n=1}^N\!\mathcal{C}_t(a_n)e^{-\frac{s(a_n+1)\beta_t}{2}}\right)\nonumber\\
&=\left( \frac{2\pi}{N} \right)^{T}\left(\prod_{t=1}^{T}\beta_t p_{1,t}\right)e^{\sum_{t=1}^{T}\frac{1}{\lambda_1p_{2,t}}}\Bigg[\sum_{n\in\mathcal{N}} \!\Bigg(\prod_{t=1}^{T}\mathcal{C}_t(a_{n})\!\Bigg) e^{-s\frac{(a_{n}+1)}{2} \sum_{t=1}^{T} \beta_t} \Bigg],
\end{align}
where $n_t$ denotes the index of the $n$-th term of $\sum_{n=1}^N\mathcal{C}_t(a_n)e^{-\frac{s\beta_t(a_n+1)}{2}}$ for a given $t$ \cite{abate-laplace},
and $\psi=\{T,T+1,\cdots,NT\}$ with $|\psi|=(N-1)T+1$ and $\mathcal{N} \triangleq \left\{n_t \big|\left(\sum_{t=1}^{T}n_t \right)\in \psi\right\}$.
Using the inverse Laplace transform to $\widehat{f}_Z(s)$ \cite{abate-laplace}, we have
\begin{align}\label{eq:16}
f_Z(z)\!\approx \!
\left(\!\frac{2\pi}{N}\right)^{T}\!\!\!\left(\prod_{t=1}^{T}\beta_t p_{1,t}\!\right)\!\!\!\left[\sum_{\sum_{t=1}^{T}n_t\in\psi}\prod_{t=1}^{T}\mathcal{C}_t(a_{n_t})\frac{1}{z}e^{\sum\limits_{t=1}^{T}(\lambda_1p_{2,t})^{-1}}\!\!\left(\!\sum_{m = 1}^{M} w_m 2^{-\frac{m}{2z}\!\sum\limits_{t=1}^{T}\beta_t(a_{n_t}+1)} \right)\!\right]\!.
\end{align}
%
%
Therefore, the outage probability of user 1 is given by $\Pout_{1,T} = \Pr(Z\le\gamma_1)$ and thus can be approximated by \eqref{eq:ana_outage_user1}. This completes the proof.
\end{IEEEproof}

\begin{rem}\label{rem:hign_snr_user1}
In the high SNR regime, i.e., $\frac{p_{1,t} + p_{2,t}}{\sigma_1^2} \gg 1$ \cite{Ding2014}, the received SINR at user 1 during $t$-th HARQ-CC round degrades into 
\begin{align}
\SINR_{1,t}^{x_1} = \frac{p_{1,t}h_{1,t}\lambda_1}{p_{2,t}h_{1,t}\lambda_1 + 1}
\approx \frac{p_{1,t}}{p_{2,t}}.
\end{align}
In this case, if $\frac{p_{1,t}}{p_{2,t}} \ge \gamma_1$, user 1 can successfully decode its information with a probability approaching 1, and thus the outage probability $\Pout_{1,t} \rightarrow 0$.
\end{rem}


To further attain more insights of the derived outage probability in \eqref{eq:ana_outage_user1},  the diversity order of user 1 is given in the following lemma.
\begin{thm}
 The diversity order of user 1 is equal to the number of retransmissions, i.e., $T$.
\end{thm}
\begin{IEEEproof}
Firstly, define the diversity order as follow \cite{cai2018-harq}:
\begin{align} \label{eq:diversity_order}
\mathcal{D} = - \lim_{\rho \rightarrow \infty} \frac{\log(\Pout(\rho))}{\log(\rho)}
\end{align}

Then, by defining $\rho_{1,t} = a_{1,t} \rho$ and $\rho_{2,t} = a_{2,t} \rho, ~\forall t$, the outage probability of user 1 can be rearranged as
\begin{align} \label{eq:ana_outage_user1_tho}
\Pout_{1,T} &\approx \mu\left(\rho^T\prod_{t=1}^{T}\beta_ta_{1,t}\sigma_1^2\right)e^{\rho^{-1}\lambda_1^{-1}\sum\limits_{t = 1}^{T}\frac{1}{a_{2,t} \sigma_2^2}}\left(\sum_{n\in\mathcal{N}}\left(\prod_{t=1}^{T}\mathcal{C}_t(a_n)\right)\right)
\end{align}
where  $$\mu \triangleq 2^{T+1}\ln 2\left(\frac{\pi}{N}\right)^{T+1}\sum_{n\in\mathcal{N}}\sum_{m=1}^{M}\sum_{k=1}^{N} \left({\frac{w_m\sqrt{1-a_{k}^2}}{\gamma_1(1+a_k)} } 2^{\frac{m(1+a_n)}{\gamma_1(1+a_{k})} -\sum_{t=1}^{T} \beta_t}\right)$$
and $\mathcal{C}_t(a_n)$ can be rewritten as
\begin{align} \label{eq:c_diversity}
\mathcal{C}_t(a_n)=\frac{\sqrt{1-a_n^2}}{\lambda_1 \rho^2(2a_{1,t} \sigma_1^2-a_{2,t}\sigma_2^2\beta_t(a_n+1))^2}e^{-\frac{2a_{1,t}}{a_{2,t}\rho(2a_{1,t}\sigma_1^2-a_{2,t}\sigma_2^2\beta_t(a_n+1))\lambda_1}},
\end{align}

The diversity order of user 1 is given by
\begin{align}
\mathcal{D}_1 &\!= -\lim_{\rho \rightarrow \infty} \frac{\log\left(\mu\left(\rho^T\prod_{t=1}^{T}\beta_ta_{1,t}\sigma_1^2\right)e^{\rho^{-1}\lambda_1^{-1}\sum\limits_{t = 1}^{T}\frac{1}{a_{2,t} \sigma_2^2}}\left(\sum_{n\in\mathcal{N}}\prod_{t=1}^{T}\mathcal{C}_t(a_n)\right) \right)}{\log(\rho)}  \nonumber\\
&= -\lim_{\rho \rightarrow \infty}\frac{\log(\rho^T)}{\log(\rho)} \!-\!\lim_{\rho \rightarrow \infty}\! \frac{\log\left(\!\mu\prod\limits_{t=1}^{T}\beta_ta_{1,t}\sigma_1^2 e^{\rho^{-1}\lambda_1^{-1}\sum\limits_{t = 1}^{T}\frac{1}{a_{2,t} \sigma_2^2}}\right)}{\log(\rho)} \!-\! \lim_{\rho \rightarrow \infty} \frac{\log\left(\sum\limits_{n\in\mathcal{N}}\prod\limits_{t=1}^{T}\mathcal{C}_t(a_n)\right)}{\log(\rho)} \nonumber\\
&= \!-T \!-\! \lim_{\rho \rightarrow \infty}\! \frac{\log(\frac{1}{\rho^{2T}})}{\log(\rho)}\!-\!\lim_{\rho \rightarrow \infty} \! \frac{\log\left(\sum\limits_{n\in\mathcal{N}}\prod\limits_{t=1}^{T}\frac{\sqrt{1-a_n^2}}{\lambda_1 (2a_{1,t} \sigma_1^2-a_{2,t}\sigma_2^2\beta_t(a_n+1))^2}e^{-\frac{2a_{1,t}}{a_{2,t}\rho(2a_{1,t}\sigma_1^2-a_{2,t}\sigma_2^2\beta_t(a_n+1))\lambda_1}}\right)}{\log(\rho)} \nonumber\\
&= -T+2T = T.
\end{align}
It can be seen that the diversity order of user 1 is $T$. This completes the proof.
\end{IEEEproof}

\subsubsection{Outage Probability of User 2}
The outage probability of user 2
is determined by the joint probability of accumulated SINRs, i.e., $\sum_{\ell=1}^{t}\SINR_{2,\ell}^{x_1}$ and $\sum_{\ell=1}^{t}\SNR_{2,\ell}^{x_2}$, and one cannot attain its closed-form expression \cite{choi-tcom-2016}. In view of this, we will approximate it in order to derive the outage probability in a closed-form expression.
In the high SNR regime, $\Pout_{2,T}$ can be approximated by
\begin{align} \label{eq:outage_approx_user2}
\Pout_{2,T} \approx 1 - \Pout \left(\sum_{t = 1}^T \frac{p_{1,t}}{p_{2,t}} \ge \gamma_1, \sum_{t = 1}^{T} p_{2,t} h_{2,t} \lambda_2 \ge \gamma_2 \right).
\end{align}
So the outage constraint of user 2 can be conservatively replaced by
\begin{subequations}
	\begin{align}
	&\sum_{t = 1}^T \frac{p_{1,t}}{p_{2,t}} \ge \gamma_1, \label{eq:outage_approx_eq2} \\
	&\Pout \left(\sum_{t = 1}^{T} p_{2,t} h_{2,t} \lambda_2 < \gamma_2 \right) \le \delta_2.  \label{eq:outage_approx_eq1}
	\end{align}
\end{subequations}

Now, the remaining is to explicitly derive the outage probability in \eqref{eq:outage_approx_eq1}.
To begin with, by letting $X_t = p_{2,t}h_{2,t}\lambda_2$ and $X \triangleq \sum_{t=1}^TX_t$,
then its moment-generating function of $X$ can be written as
\begin{subequations}
	\begin{align}
	\mathcal{M}_X(s) & \!= \int_{0}^{+\infty} f_X(x)e^{sx } dx \\
   &\overset{(a)}{=}\prod_{t=1}^T\int_{0}^{+\infty} f_{X_t}(x_t)e^{sx_t}dx_t\\
	& \!= \left(\prod_{t=1}^{T}\frac{1}{p_{2,t}\lambda_2}\right) \prod_{t=1}^{T}\int_{0}^{+\infty} e^{-\frac{x_t(1-s\lambda_2p_{2,t})}{p_{2,t}\lambda_2}} dx_t \\
	& \!=\! \left(\prod_{t=1}^{T}\frac{1}{p_{2,t}\lambda_2}\right) \prod_{t=1}^{T} \frac{p_{2,t}\lambda_2}{1 - s\lambda_2p_{2,t}} \\
	& = \prod_{t=1}^{T}\frac{1}{1 - s\lambda_2p_{2,t}}, \label{eq:MGF_user2}
	\end{align}
\end{subequations}
where $f_X(x)$ is the PDF of $X$, and (a) is due to the fact that $X_t$, $t = 1,\cdots,T$, are statistically independent.
Applying the inverse Laplace transform to \eqref{eq:MGF_user2}, we have
\begin{subequations}
	\begin{align}
	F_X(x) &= L^{-1} \left(L\left(\int_0^{+\infty} f_X (x) dx\right)\right) \\
	& = L^{-1} \left(\frac{1}{s}M_X(-s)\right) \\
	& = L^{-1} \left(\frac{1}{s}\prod_{t=1}^{T} \frac{1}{1+s\lambda_2p_{2,t}}\right) \label{eq:Gaver-Stehfest1}\\
	& \approx  \sum_{m = 1}^{M} w_m \prod_{t = 1}^{T} \frac{1}{1+\frac{m \lambda_2 \ln 2}{x} p_{2,t}}, \label{eq:outage_probability_user2}
	\end{align}
\end{subequations}
where \eqref{eq:outage_probability_user2} is obtained based on the Gaver-Stehfest procedure \cite{abate-laplace}, and $w_m$ is the same as that in \eqref{eq:16}.
Thus the outage probability of user 2 can be given by $\Pout_{2,T} = F_X(\gamma_2)$.

Here we claim the tightness of the approximation by \eqref{eq:outage_approx_eq1} for a special case of no retransmission, i.e., $T = 1$, in the following lemma.
\begin{lem}\label{lem:outage_tightness}
When $T = 1$, if the power allocation satisfies $\frac{p_{1}}{p_{2}} \ge \frac{\gamma_1 + \gamma_1\gamma_2}{\gamma_2}$, the outage probability of user 2 in \eqref{eq:outage_2} is strictly equal to that in \eqref{eq:outage_approx_eq1}.
\end{lem}
\begin{IEEEproof}
	Firstly, we have 
   \begin{align}
   \frac{p_{1}}{p_{2}} \ge \frac{\gamma_1 + \gamma_1 \gamma_2}{\gamma_2} 
     \Leftrightarrow 	\frac{\gamma_1}{\left(p_{1}-\gamma_1p_{2}\right)\lambda_2} \le \frac{\gamma_2}{p_{2}\lambda_2}. \label{eq:outage_equivalence}
	\end{align}
		Then, recall the outage probability of user 2 when $T=1$ and rewrite it as follows:
	\begin{subequations}
		\begin{align}
		\Pout_{2,1}  & = 1 - \Pout \left( \frac{p_{1}h_{2}\lambda_2}{p_{2}h_{2}\lambda_2 + 1} \ge \gamma_1,  p_{2} h_{2}\lambda_2 \ge \gamma_2\right)\\
		& = 1 - \Pout\left(h_{2} \ge \frac{\gamma_1}{\left(p_{1}-\gamma_1p_{2}\right)\lambda_2}, h_{2} \ge \frac{\gamma_2}{p_{2}\lambda_2}\right),\\
		& \overset{(a)}{=} 1 - \Pout\left( h_{2} \ge \frac{\gamma_2}{p_{2}\lambda_2}\right), \label{eq:29b}
 		\end{align}
	\end{subequations}
where (a) is obtained by assuming $ \frac{\gamma_1}{\left(p_{1}-\gamma_1p_{2}\right)\lambda_2} \le \frac{\gamma_2}{p_{2}\lambda_2}$, which is the same as the right-hand side of \eqref{eq:outage_equivalence}. This indicates that if the power allocation satisfies the left-hand side of \eqref{eq:outage_equivalence}, the the outage probability of user 2 in \eqref{eq:outage_2} is strictly equal to that in \eqref{eq:outage_approx_eq1}, i.e., \eqref{eq:29b}  when $T = 1$.
This completes the proof.
\end{IEEEproof}



\begin{lem}
	The diversity order of user 2 is equal to the number of retransmission rounds, i.e., $T$.
	\end{lem}
\begin{IEEEproof}
Similar to that for user 1, we first rewrite the outage probability of user 2 as
\begin{align}
\Pout_{2,T} = \sum_{m = 1}^{M} w_m \prod_{t = 1}^{T} \frac{1}{1+ A_{2,m,t} \rho_{2,t}}
\end{align}
where $A_{2,m,t} = \frac{m \lambda_2 \sigma_2^2 \alpha_{2,t}\ln 2}{\gamma_2}$. By letting $\rho_{2,t} = a_{2,t} \rho$ and focusing on high SNR, the diversity order of user 2 can be given by
\begin{subequations}
	\begin{align}
	\mathcal{D}_2 &= -\lim_{\rho \rightarrow \infty} \frac{\log\left(\sum_{m = 1}^{M} w_m \prod_{t = 1}^{T} \frac{1}{A_{2,m,t} a_{2,t} \rho}\right)}{\log(\rho)} \\
	 &=  -\lim_{\rho \rightarrow \infty} \frac{\log\left(\frac{1}{\rho^T}\right)}{\log\left(\rho\right)}  - \lim_{\rho \rightarrow \infty} \frac{\log\left(\prod_{t = 1}^{T} \frac{1}{a_{2,t}}\right)}{\log\left(\rho\right)}  - \lim_{\rho \rightarrow \infty} \frac{\log\left(\sum_{m = 1}^{M} w_m \prod_{t = 1}^{T} \frac{1}{A_{2,m,t}}\right)}{\log\left(\rho\right)} \\
	 & = T
	\end{align}
\end{subequations}
It can be seen that user 1 can achieve the full diversity gain in the HARQ-CC enabled NOMA systems.
\end{IEEEproof}

On the roles of the derived outage probabilities in \eqref{eq:ana_outage_user1} and \eqref{eq:outage_probability_user2}  we have the following remark. 
\begin{rem}
	 Firstly, with \eqref{eq:ana_outage_user1} and \eqref{eq:outage_probability_user2}, the complex outage probabilities characterized in \eqref{eq:outage_users} are explicitly approximated with high accuracy. Furthermore, based on the derived outage probabilities, the diversity orders are analyzed and both the users can achieve the full diversity gain. Secondly, in the following sections, we will use \eqref{eq:ana_outage_user1} and \eqref{eq:outage_probability_user2} to optimize the power allocation in order to improve the system performance.
\end{rem}

Up to now, the outage probabilities of users are expressed in closed-form as in \eqref{eq:ana_outage_user1} and \eqref{eq:outage_probability_user2}. However, it is still difficult to solve problem \eqref{p:power_min_ARQ}, since the outage probability constraints are nonconvex and extremely complicated, especially for user 1. One possible approach to treat such a nonconvex problem is via exhaustive search. By using the exhaustive search method, the transmit power can be discretized into $L$ levels, and one needs to exhaustively search over all the $L^{2T}$ possible combinations of power allocation to find the one that minimizes the average transmit power. For example, for $T = 3$ and $L = 100$, it requires to search over $10^{12}$ possible combinations, which is computationally prohibitive in practice. Thus this approach is viable only when $T$ is small.
To resolve this challenge, in the following, a suboptimal solution will be developed by carrying out approximations to problem \eqref{p:power_min_ARQ}. Note that the exhaustive search based method will be used in section \ref{sec:simulation} as a benchmark to evaluate the accuracy of the approximation-based algorithm.


\subsection{Adaptive Power Allocation Design} \label{sec:power_min}
Recall that the outage probability of user 2 can be approximated in the high SNR region and this approximated expression shows good tightness in the low SNR region, as presented in Fig. \ref{fig:outage_user2}.
To avoid high computational complexity to directly solve problem \eqref{p:power_min_ARQ}, we approximate the outage of user 1 in a similar way to that of user 2. Specifically, at high SNR region (recall Remark \ref{rem:hign_snr_user1}), if the power allocation satisfies $\frac{p_{1,t}}{p_{2,t}} \ge \gamma_1$, we can approximately believe that user 1 can decode its information with probability approaching 1 (i.e., $\Pout_{1,t} \approx 0,~\forall t$) at each HARQ-CC round, and the retransmission probability in \eqref{eq:retransmission_probability} is degraded into $\hat{\Pout}_t =  \Pout_{2,t-1}$.
Note that with the constraint $\frac{p_{1,t}}{p_{2,t}} \ge \gamma_1,~ \forall t$, the requirement of the power allocation in \eqref{eq:outage_approx_eq2} is automatically satisfied.
Then, by inserting the outage probability of $ \Pout_{2,t-1}$ in \eqref{eq:outage_probability_user2}, the average power minimization problem can be approximated by
\begin{subequations} \label{p:power_min_ARQ_approx}
	\begin{align}
	\min_{{\bf p_1}, {\bf p_2}} ~& p_{1,t} + p_{2,t} +\sum_{t=2}^{T} (p_{1,t} + p_{2,t}) \sum_{m = 1}^{M} w_m \prod_{\ell = 1}^{t-1} \frac{1}{1+g_m p_{2,\ell}}\\
	\st ~~& \sum_{m = 1}^{M} w_m \prod_{t = 1}^{T} \frac{1}{1+g_m p_{2,t}} \le \delta_2,  \label{eq:33b}\\
	& \frac{p_{1,t}}{p_{2,t}} \ge \gamma_1,   \forall  t,\label{eq:con_frac}\\
	& p_{1,t} \ge 0, ~ p_{2,t} \ge 0, \forall  t,\\
	&p_{1,t} + p_{2,t} \le P_{\max},  \forall t,
	\end{align}
\end{subequations}
where $g_m = \frac{m \lambda_2 \ln 2 }{\gamma_2}$. Notice that problem \eqref{p:power_min_ARQ_approx} is still nonconvex and complicated due to the outage probabilities of users both in the objective function and the constraint \eqref{eq:33b}, which are the multiplications of multiple variables. Generally, one can resort to the geometric programming (GP) method to handle such nonconvex problems, e.g., \cite{larsson-tcomm-2013}. However, the GP method is not valid to problem \eqref{p:power_min_ARQ_approx} since the parameters ($w_m$'s) in the outage probability expression can be negative, i.e., $w_2,w_4,w_6,\cdots$.
In what follows, we will show how to use a convex approximation method to solve problem \eqref{p:power_min_ARQ_approx} efficiently.

\noindent{\bf Successive Convex Approximation (SCA) based Algorithm :}
First,
{{by introducing auxiliary variables $u_1$ and $u_2$, the objective function can be recast as}}
\begin{subequations}
	\begin{align}
	\min_{{\bf p_1}, {\bf p_2},u_1,u_2} ~~& p_{1,1} + p_{2,1} + u_1 + u_2\\
	\st ~~&\sum_{t = 2}^T p_{1,t} \sum_{m = 1}^{M} w_m \prod_{\ell = 1}^{t-1} \frac{1}{1+g_m p_{2,\ell}} \le u_1,\\
	&\sum_{t = 2}^T p_{2,t} \sum_{m = 1}^{M} w_m \prod_{\ell = 1}^{t-1} \frac{1}{1+g_m p_{2,\ell}} \le u_2.
	\end{align}
\end{subequations}
 Then we define
\begin{align} \label{eq:CoV}
\exp(x_{m,t}) &\triangleq \frac{1}{1+g_m p_{2,t}}, ~~
\exp(y_t) \triangleq p_{1,t}, ~~
\exp(z_t) \triangleq p_{2,t}.
\end{align}
Thus problem \eqref{p:power_min_ARQ_approx} can be rewritten as
\begin{subequations} \label{p:power_min_CoV}
	\begin{align}
	\min_{\substack{\{x_{m,t}\},\{y_t\},\{z_t\},u_1,u_2}} ~~&~ \exp(y_1) + \exp(z_1) + u_1 + u_2 \label{eq:power_min_CoV1}\\
	\st ~~~~& \sum_{t = 2}^T \sum_{m=1}^{M} w_m \exp\left(y_t + \sum_{\ell=1}^{t-1}x_{m,\ell}\right) \le u_1,  \label{eq:power_min_CoV2}\\
	 & \sum_{t = 2}^T  \sum_{m=1}^{M} w_m \exp\left(z_t + \sum_{\ell=1}^{t-1}x_{m,\ell}\right) \le u_2,  \label{eq:power_min_CoV21}\\
	& \sum_{m=1}^{M} w_m \exp\left(\sum_{t=1}^{T} x_{m,t}\right) \le \delta_2,\label{eq:power_min_CoV3}\\
	& \exp(x_{m,t}) + g_m \exp(x_{m,t} + z_t)=1, \label{eq:power_min_CoV4}\\
	& y_t-z_t \ge \log(\gamma_1), \label{eq:power_min_CoV5}\\
	& \exp(y_t) + \exp(z_t) \le P_{\rm max}. \label{eq:power_min_CoV7}
	\end{align}
\end{subequations}

Note that problem \eqref{p:power_min_CoV} is still non-convex.
However, compared to problem \eqref{p:power_min_ARQ_approx}, the structure of problem \eqref{p:power_min_CoV} allows us to employ the convex approximation method \cite{XU-TSP-2017} to approximate it by a convex one.
As a result, problem \eqref{p:power_min_CoV} can be approximately solved in an iterative manner.
Specifically, in the $r$-th iteration, at a feasible point $\left\{x_{m,t}^{(r-1)},y_t^{(r-1)},{z}_t^{(r-1)}\right\}$ of problem \eqref{p:power_min_CoV},
one needs to do linear approximation to the nonconvex constraints in \eqref{p:power_min_CoV}. Specifically, constraint \eqref{eq:power_min_CoV2} and  \eqref{eq:power_min_CoV21} can be respectively approximated by
\begin{subequations}
	\begin{align}
	&\!\!\!\!\sum_{t = 2}^T\!\Bigg(\!\sum_{b = 1}^{M/2} w_{2b\!-\!1} \exp\left({{y}_t \!+\! \sum\limits_{\ell=1}^{t\!-\!1}x_{2b\!-\!1,\ell}}\right)
	\!\!+\! \sum_{b = 1}^{M/2} \!w_{2b} \!\exp\left({{y}_t^{(r\!-\!1)}\!+\!\sum\limits_{\ell=1}^{t\!-\!1}x_{2b,\ell}^{(r\!-\!1)}} \right) \nonumber \\
	&~~~~~~~~~~~~~~~~~~~~~~~~~~~~~~~~~~~~~~~~\times \left(\!1\!+\!{y}_t - {y}_t^{(r-1)} +\sum_{\ell=1}^{t-1} x_{2b,\ell}  - \sum_{\ell=1}^{t-1} x_{2b,\ell}^{(r-1)} \right)\!\!\Bigg)  \le u_1, \label{eq:objective_approx1} \\
	&\!\!\!\!\sum_{t = 2}^T\!\Bigg(\!\sum_{b = 1}^{M/2} w_{2b\!-\!1} \exp\left({{z}_t \!+\! \sum\limits_{\ell=1}^{t\!-\!1}x_{2b\!-\!1,\ell}}\right)
	\!\!+\! \sum_{b = 1}^{M/2} \!w_{2b} \!\exp\left({{z}_t^{(r\!-\!1)}\!+\!\sum\limits_{\ell=1}^{t\!-\!1}x_{2b,\ell}^{(r\!-\!1)}} \right) \nonumber \\
	&~~~~~~~~~~~~~~~~~~~~~~~~~~~~~~~~~~~~~~~~\times \left(1\!+\!{z}_t - {z}_t^{(r-1)} +\sum_{\ell=1}^{t-1} x_{2b,\ell}  - \sum_{\ell=1}^{t-1} x_{2b,\ell}^{(r-1)} \right)\!\!\Bigg)  \le u_2,\label{eq:objective_approx2}
	\end{align}
\end{subequations}
which are convex upper-bound of the left-hand side of constraints \eqref{eq:power_min_CoV2} and \eqref{eq:power_min_CoV21} respectively by applying the first-order Taylor approximation. The other nonconvex constraints can be treated similarly. Consequently, we obtain the following approximation of problem \eqref{p:power_min_CoV}
\begin{subequations} \label{p:power_min_sca}
	\begin{align}
	&\min_{\substack{\{x_{m,t}\},\{y_t\},\{{z}_t\},u_1,u_2}}  ~ \exp\left(y_1 \right) + \exp\left(z_1\right) + u_1 + u_2 \\
	\st ~~&\eqref{eq:objective_approx1},\eqref{eq:objective_approx2}  \nonumber \\
	&\sum_{b = 1}^{M/2} \!w_{2b-1} \exp\left({\sum\limits_{t=1}^{T}x_{2b\!-\!1,t}}\right) \!\!+\!\sum_{b= 1}^{M/2} \!w_{2b} \!\exp\left({ \sum\limits_{t=1}^{T}x_{2b,t}^{(r\!-\!1)}} \right)\left(1\!+\! \sum_{t=1}^{T}x_{2b,t} \!-\! \sum_{t=1}^{T}x_{2b,t}^{(r-1)} \!\right) \!\le\! \delta_2, \\
	&\exp\left({x_{m,t}^{(r\!-\!1)}}\right)\left(1\!+\!x_{m,t} \!-\! x_{m,t}^{(r\!-\!1)}\right) \!+\! g_m \exp\left({x_{m,t}^{(r\!-\!1)} \!+\! z_{t}^{(r\!-\!1)}}\right)\left(1\!+\!x_{m,t} \!+\! {z}_t \!-\! x_{m,t}^{(r\!-\!1)} \!-\! z_{t}^{(r\!-\!1)}\right) = 1,\\
	&y_t-z_t \ge \log(\gamma_1),\\
	&\exp\left({y_t}\right) + \exp\left({z}_t\right) \le P_{\rm max},
	\end{align}
\end{subequations}
which is a convex optimization problem and thus can be efficiently solved by standard convex solvers, e.g., \texttt{CVX} \cite{cvx}. After solving problem \eqref{p:power_min_sca}, {{if the gap between the optimal values in two successive iterations is larger than a desired accuracy $\epsilon$}}, one needs to update $\left\{x_{m,t}^{(r)},y_t^{(r)},{z}_t^{(r)}\right\}$ with the obtained solutions, and then solve the problem in $(r+1)$-th iteration. Finally, the power allocation $\{p_{1,t}^{\star},p_{2,t}^{\star}\}$ can be obtained based on \eqref{eq:CoV}. The procedure of the SCA based algorithm is outlined in Algorithm \ref{alg:sca}, and in fact, we can draw the following proposition.

\begin{algorithm}[!t]\small
	\caption{~SCA based algorithm for solving \eqref{p:power_min_CoV} }\label{alg:sca}
	\begin{algorithmic}[1]
		\STATE {{\bf Initialization:} Set $r = 0$, given a set of feasible $\left\{x_{m,t}^{(0)},y_t^{(0)},{z}_t^{(0)}\right\}$, and desired accuracy $\epsilon$.}
		\STATE {\bf repeat}
		\STATE {\quad Update $\left\{x_{m,t}^{(r)},y_t^{(r)},{z}_t^{(r)}\right\}$ by solving problem \eqref{p:power_min_sca}.}
		\STATE {\quad Set $r \leftarrow r+1$.}
		\STATE {{\bf until} the power gap between two successive iteration is smaller than $\epsilon$.}
		\STATE {{\bf Output:} $\left\{\!x_{m,t}^{(r)},y_t^{(r)}\!,{z}_t^{(r)}\right\}$, and calculate $\{p_{1,t}^{\star},p_{2,t}^{\star}\}$ based on \eqref{eq:CoV}.}
	\end{algorithmic}

\end{algorithm}

\begin{prop}\label{prop1}
	The proposed algorithm can continuously decrease the power consumption gap between two successive iterations and guarantee the generated power consumption sequence converges to at least a stationary point of problem \eqref{p:power_min_ARQ_approx}.
\end{prop}

\begin{proof}
	The proof of Proposition \ref{prop1} is similar as that in \cite[Theorem 1]{Li-2013-tsp}, thus we omit it here.
	\hfill $\blacksquare$
\end{proof}

%

\begin{rem}
It is important to point out that, as a counterpart, the transmission round minimization problem under outage constraints can also be inherently solved by Algorithm \ref{alg:sca}. In particular, the corresponding optimization problem can be formulated as
\begin{subequations}\label{p:round_min}
	\begin{align}
	\min_{{\bf p_1}, {\bf p_2},\hat{t}}~~~& \hat{t} \\
	\st ~~~
	& \Pout_{k, \hat{t}}\le \delta_k, {\rm for}~ k =1,2.  \label{eq:round_min_outage}\\
	& p_{1,t} \ge 0, ~ p_{2,t} \ge 0,  \forall t \in \{1,...,\hat{t}\} \\
	& p_{1,t} + p_{2,t} \le P_{\max}, \forall t \in \{1,...,\hat{t}\} \label{eq:round_min_power}\\
	& \hat{t} \in \{1,2,...,T\}.
	\end{align}
\end{subequations}

Note that, although one can plug the outage probabilities in \eqref{eq:ana_outage_user1} and \eqref{eq:outage_probability_user2} into problem \eqref{p:round_min}, the problem still doesn't have an unified expression as the outage probabilities have different expressions with different $\hat{t}$.
Notice that the outage probabilities of users are monotonically decreasing with $\hat{t}$. Thus it allows us to efficiently search $\hat{t}$ by using the bisection search method \cite{Boyd-2009}.
Specifically, for each given $\hat{t}$, one only needs to check the feasibility of problem \eqref{p:round_min}. If it is feasible, one needs to decrease $\hat{t}$; otherwise $\hat{t}$ should be increased.
\end{rem}

In this Section, we focus on the problems in the two-user case. However, in a practical communication system, there are randomly deployed with multiple users. Thus, it motivates us to investigate the system design of the multi-user scenario. In the following section, we will consider the power minimization design in a multi-user NOMA system, and show how the derived results in the two-user case can be extended to the multi-user case.

\section{User Paring and Power Allocation in Multi-user Cases} \label{sec:multi-user}

In this section, we consider a more general scenario where a single-antenna BS tries to serve $2K$ single antenna users, as shown in Fig. \ref{fig:sys_model2}. In particular, the users are divided into two groups, denoted as $\kset_1$ and $\kset_2$, in terms of their distances to the BS.
We assume that the $K$ cell center users (CUs) in $\mathcal{K}_1 =\left\{ {1, 2, \cdots, K} \right\}$ are randomly distributed within a circle district near to the BS with radius $R_c$ and the other $K$ cell edge users (EUs) in $\mathcal{K}_2=\left\{ {K+1, K+2, \cdots, 2K} \right\}$ are located in a ring district with inner radius $R_c$ and outer radius $R_e$. To benefit the advantage of NOMA, each CU is paired with one EU to perform non-orthogonal transmission.
Denote by $\bf V = [v_{i,j}]$ the user paring matrix, where $v_{ij} = 1$ denotes that user $i$ in group $\kset_1$ and user $j$ in group $\kset_2$ are paired to perform NOMA transmission, otherwise $v_{ij} = 0$.
Moreover, we assume that different pairs are allocated with different orthogonal resource blocks, such as orthogonal frequency division multiplexing (OFDM) subcarriers, to avoid inter-pair interference.
\begin{figure}[!t]
	\centering
	\includegraphics[width=0.7\linewidth]{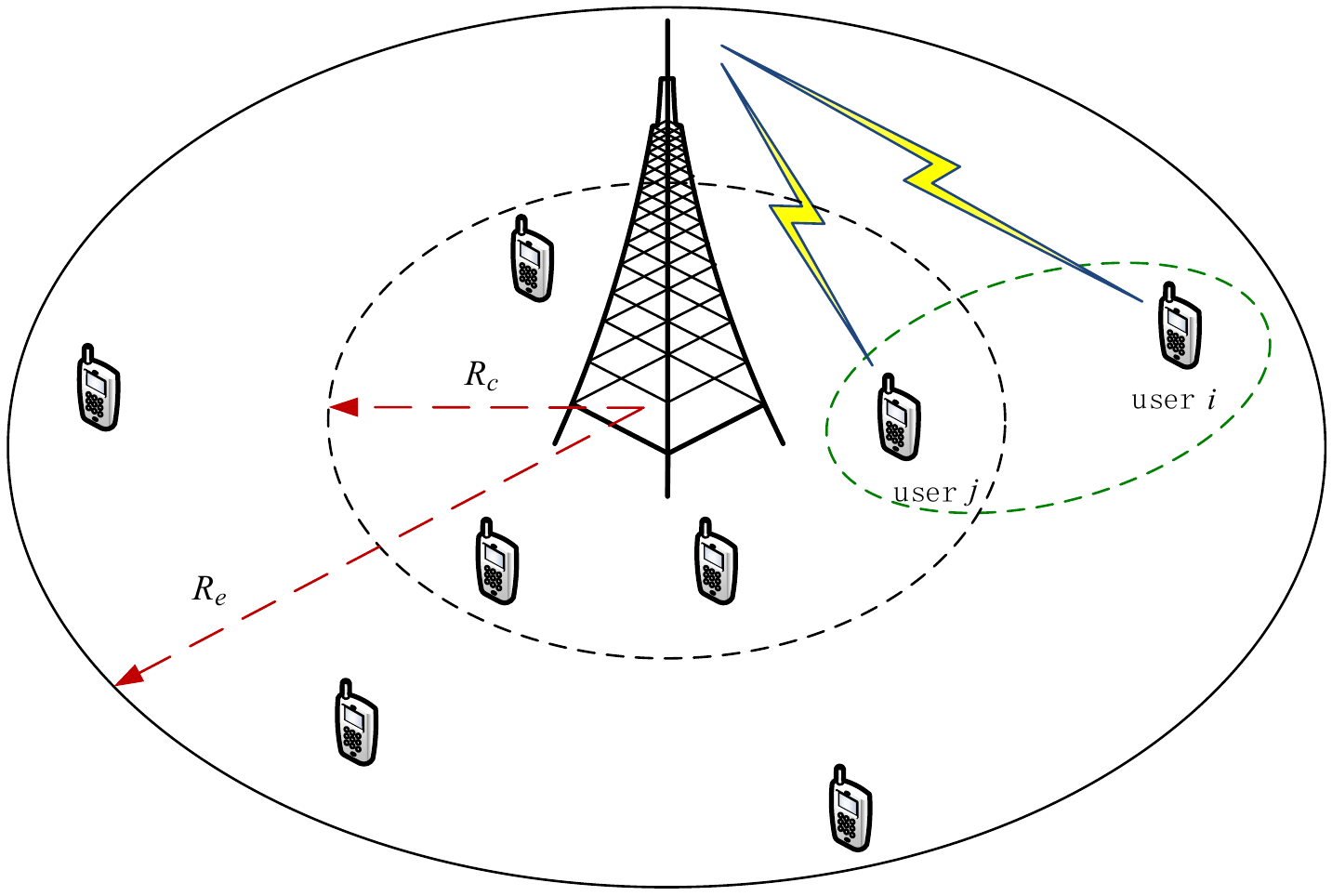}\\
	\caption{An illustration of user paring in a multi-user HARQ-CC enabled NOMA system.} \label{fig:sys_model2}
\end{figure}
Similar to \eqref{eq:averge_power_2}, we define $P_{ij,{\rm avg}}$ as the average transmit power of a potential user pair of use $i$ and user $j$, which satisfies that
\begin{align}
P_{ij,{\rm avg}} = p_{i,1} + p_{j,1} + \sum_{t = 2}^{T} (p_{i,t}+p_{j,t})\hat{\Pout}_{ij,t},
\end{align}
where $\hat{\Pout}_{ij,t} = \Pout_{i,t-1} + \Pout_{j,t-1} - \Pout_{i,t-1} \Pout_{j,t-1}$ denotes the retransmission probability of user $i$ and user $j$ at the $t$-th retransmission round.
Therefore, the average transmit power minimization problem in the multi-user scenario can be formulated as
\begin{subequations}\label{p:power_min_ARQ_mu}
	\begin{align}
	\min_{{\bf P_1},{\bf P_2},{\bf V}} \quad & \sum_{i \in \kset_1}\sum_{j \in \kset_2} v_{ij} P_{ij,{\rm avg}}  \\
	\st ~~~~& \Pout_{i, T}\le \delta_i, ~ \Pout_{j, T}\le \delta_j, \forall i \in \kset_1 , j \in \kset_2,\label{eq:outage_mu1}\\
	& \sum_{j \in \kset_2} v_{ij} \left(p_{i,t} + p_{j,t}\right) \le P_T, ~ \forall i \in \mathcal{K}_1, t, \label{eq:power_mu}\\
	& p_{i,t} \ge 0, ~ p_{j,t} \ge 0, \forall i \in \mathcal{K}_1, j \in \mathcal{K}_2, t, \\
	& \sum_{j \in \kset_2} v_{ij} = 1,  \forall i \in \mathcal{K}_1, \label{eq:paring_mu_1}\\
	& \sum_{i \in \kset_1} v_{ij} = 1,  \forall j \in \mathcal{K}_2, \label{eq:paring_mu_2}\\
	& v_{ij} \in \{0,1\}, \forall i \in \mathcal{K}_1, j \in \mathcal{K}_2\label{eq:paring_mu_3}
	\end{align}
\end{subequations}
where ${\bf P_1} = [p_{i,t}]_{K \times T}$, ${\bf P_2} = [p_{j,t}]_{K \times T}$, ${\bf V} = [v_{i,j}]_{K \times K}$ and $P_T = \frac{P_{\max}}{K}$ denotes the maximum power that can be allocated to each user-pair. \eqref{eq:outage_mu1} describe the outage constraints of users in $\mathcal{K}_1$ and $\mathcal{K}_2$ respectively. \eqref{eq:power_mu} denotes the transmit power constraint on each user-pair and the user-paring constraints in \eqref{eq:paring_mu_1} and \eqref{eq:paring_mu_2} indicate that each CU should be paired with one EU and each EU can only be paired with one CU, respectively.

Problem \eqref{p:power_min_ARQ_mu} is a mixed integer programming problem and is NP-hard in general \cite{Liu2014,Di2016,liang2017}, i.e., it is polynomial time unsolvable.
Hence, to solve problem \eqref{p:power_min_ARQ_mu} efficiently, we propose a suboptimal algorithm with much reduced complexity, where
problem \eqref{p:power_min_ARQ_mu} is decoupled into two subproblems, i.e., user pairing and power allocation. In what follows, we will show how to use matching theory to solve the user paring problem.

\subsection{User Paring with Matching Theory}
To use matching theory to solve the user-pairing problem, we first model the user pairing problem as a two-sided one-to-one matching problem, in which each CU from $\kset_1$ can match with one EU from $\kset_2$ and the users in the same user pair can share the same resource to improve their performance with NOMA transmission. More specifically, we give the the definition of two-sided one-to-one matching as below:

{\emph{Definition 1}}:  \emph{(Two-sided One-to-One Matching)} Consider CUs and EUs as two disjoint sets, $ \kset_1=\left\{ {1, 2, \cdots, K} \right\}$ and $ \kset_2=\left\{ {K+1, K+2, \cdots, 2K} \right\}$, we call the the matching $\mathcal{M}(\cdot)$ between $\kset_1$ and $\kset_2$ as a \emph{two-sided one-to-one matching} if the following conditions are satisfied:
\begin{enumerate}
	\item  $\mathcal{ M}({\rm CU}_i) \subseteq {\kset_2} $, $ \forall i \subseteq \kset_1$, $\mathcal{M}({\rm EU}_j)\in {\kset_1} $,  $\forall j \in \kset_2$;
	\item  $|\mathcal{M}({\rm CU}_i)|=1 $, $\forall i \in \kset_1$, $|\mathcal{M}({\rm EU}_j)|=1 $, $\forall j \in \kset_2$;
	\item  $\mathcal{M}({\rm CU}_i) = {\rm EU}_j $ if and only if $\mathcal{M}({\rm EU}_j) = {\rm CU}_i$,
\end{enumerate}
Constraint 1) implies that each CU is matched with an EU in $\kset_2$ and each EU is matched with a CU in $\kset_1$; Constraint 2) states that each EU can be matched with only one CU and vice versa; Constraint 3) represents that if CU $i$ is matched with EU $j$, then EU $j$ should be matched with CU $i$.

In the proposed matching process, without loss of generality, we assume that the CUs act as proposers and the EUs are selectors. At the beginning of matching, each CU proposes to match with an EU according to its preference list which is built based on the consumed transmit power to satisfy the outage requirements of users {{by solving the power allocation problem}},
and can be described as follows:
\begin{subequations}
	\begin{align}
	&{\rm CU}_{PF}(i)=[{\rm EU}_{\kset_2}(1), {\rm EU}_{\kset_2}(2), \dots, {\rm EU}_{\kset_2}(j),\dots, {\rm EU}_{\kset_2}(K)], ~\forall i \in \kset_1,
	\label{eq: preferece_CUi}\\
	&{\rm EU}_{PF}(j)=[{\rm CU}_{\kset_1}(1), {\rm CU}_{\kset_1}(2), \dots, {\rm CU}_{\kset_1}(i),\dots, {\rm CU}_{\kset_1}(K)], ~\forall j \in \kset_2.
	\label{eq:preference_EUj}
	\end{align}
\end{subequations}
where the lists are sorted in an increasing order of the consumed power.
The preference lists satisfy the condition that the user who can provide the minimum consumed power in the opposite set will be ranked as the first one. This can be described as
\begin{subequations}
	\begin{align}
	{\rm CU}_{\kset_1}(i)_j &\succ {\rm CU}_{\kset_2}(i)_{j'},  \label{eq:CU}\\
	{\rm EU}_{\kset_2}(j)_i &\succ {\rm EU}_{\kset_2}(j)_{i'},\label{eq:EU}
	\end{align}
\end{subequations}
for all $i$, $j$ in $\kset_1$ and $\kset_2$ respectively. This means that CU $i$ prefers to match with EU $j$ to EU $j'$ if the total consumed power of user pair $\left<{\rm CU}_i, {\rm EU}_j\right>$ is less than than that of the user pair $\left<{\rm CU}_i, {\rm EU}_{j'}\right>$. Similarly, \eqref{eq:EU} describes that ${\rm EU}_j$ prefers ${\rm CU}_i$ to ${\rm CU}_{i'}$ if the user pair $\left<{\rm CU}_i, {\rm EU}_j\right>$ consumes less power than the user pair $\left<{\rm CU}_{i'}, {\rm EU}_j\right>$.

%
%

To better handle the minimization of the total transmit power, we introduce the swap operation between two CUs, e.g., ${\rm CU}_i$ and ${\rm CU}_{i'}$, in $\kset_1$ and the corresponding two matched EUs, e.g., $\mathcal{M}({\rm CU}_{i'})$ and $\mathcal{M}({\rm CU}_{i'})$, in $\kset_2$ in the matching process. The swap matching can be defined by
\begin{align}\label{swap operation}
	\mathcal{M}_i^{i'}=&\mathcal{M} \backslash \{\left<{\rm CU}_i,\mathcal{M}({\rm CU}_i)\right>, \left<{\rm CU}_{i'},\mathcal{M}({\rm CU}_{i'})\right>\}~ \cup \nonumber \\
	&\{\left<{\rm CU}_i,\mathcal{M}({\rm CU}_{i'})\right>, \left<{\rm CU}_{i'},\mathcal{M}({\rm CU}_{i})\right>\}
\end{align}
where ${\rm CU}_i$ and ${\rm CU}_{i'}$ switch the matched EUs while keeping other user pairs in the matching scheme invariant. Based on the swap operation, we define the swap-blocking pair as below:

{\emph{Definition 2}}: \emph{(Swap-Blocking Pair)}  Given a matching ${\mathcal M}$ and two user-pairs $\left<{\rm CU}_i, {\rm EU}_j\right>$ and  $\left<{\rm CU}_{i^{\prime}}, {\rm EU}_{j^{\prime}}\right>$.
If there exists a swap matching $\mathcal{M}_i^{i'} $ such that the 
total transmit power of the new user pairs gets a decrease, 
then the swap operation is approved, and $\left<{\rm CU}_i, {\rm EU}_{j'}\right>$, $\left<{\rm CU}_{i'}, {\rm EU}_{j}\right>$ are swap blocking pairs under the matching $\mathcal{M}_i^{i'}$.

Note that the above definition implies that there is a benefit by exchanging the matching of user-pairs $\left<{\rm CU}_i, {\rm EU}_{j}\right>$, $\left<{\rm CU}_{i'}, {\rm EU}_{j'}\right>$ and this operation will not hurt the benefit of the other user-pairs. Thus, with this new matching strategy, the total transmit power of the system can be decreased.
Based on the above definitions, the matching behaviour of the users can be described as follows. Firstly, every two user-pairs can be arranged by the BS to form a potential swap blocking pair. Then the BS will check whether these two user pairs can get a benefit by exchanging their matches. The users will keep performing approved swap operations until they reach to a stable status, which is also known as \emph{two-side exchange stable} matching. Its definition is described as below.

{\emph{Definition 3}}: {\emph{(Two-Side Exchange Stable) \cite{roth_sotomayor_1990}:}}  If a matching $\mathcal{M}$ is not blocked by any swap-blocking pair, then $\mathcal{ M}$ is a \emph{two-side exchange stable} matching.

In the following subsections, according to the above definitions, we will proposed a swap-matching based algorithm to solve the user paring problem and the performance of the proposed algorithm will be analyzed.

\begin{algorithm}[!t] \small
	\caption{Swap-matching based algorithm for user-pairing }\label{alg:matching}
	\begin{algorithmic}[1]
		\STATE Initialize $ {\mathcal{S_{\rm UNMATCH}}=\left\{1, 2, \cdots, N\right\}} $,  $\mathcal{M_{\rm MATCH}} = \phi$ and  $\mathcal{M} = \phi$.
		\STATE Initialize preference lists for CUs and EUs as $ {\rm CU}_{PF}(i), i \in \kset_1$ and ${\rm EU}_{PF}(j), j \in \kset_2$, respectively.\\
		\vspace{1mm}
		{\setlength\parindent{-1.1em} \bf Initial Matching Phase: }\vspace{1mm}
		\FOR{$i=1$ to $K$}
		\STATE Each ${\rm CU}_i$ sends matching request to its most preferred EU $j$ according to ${\rm CU}_{PF}(i)$.
		\IF {${\rm EU}_{j}$ has not been matched with any CU}
		\STATE Add user pair $\left<{\rm CU}_i, {\rm EU}_{j}\right>$ to the existing matching scheme $\mathcal{M_{\rm MATCH}}$. Remove ${\rm CU}_i$ from $\left\{ {\mathcal{S_{\rm UNMATCH}}} \right\}$.
		\ELSE
		\STATE ${\rm CU}_i$ will be matched with the unmatched ${\rm EU}_{j^{\prime}}$ according to ${\rm CU}_{PF}(i)$. \\
		\ENDIF
		\ENDFOR\\\vspace{1mm}
		{\setlength\parindent{-1.5em} \bf Swap Matching Phase: }\vspace{1mm}
		\REPEAT
		\STATE Search the matched pairs to check whether there exists swap-blocking pair.
		\IF {there is swap blocking pair}
		\STATE Swap the user pair and update the current matching scheme.
		\ENDIF
		\UNTIL There is no swap-blocking pair in the matching.
	\end{algorithmic}
\end{algorithm}

\subsection{Swap Matching based Algorithm Description}
In this subsection, we proposed a swap operations enabled matching algorithm for user pairing problem. We first initialize the preference lists for each user according to the power consumption under outage constraints and then design the swap operations enabled matching step to minimize the total transmit power. 
At the initialization step, as shown in Step 1 of Alg. \ref{alg:matching}, the set $\mathcal{S_{\rm UNMATCH}}$ is defined to record CUs who have not been matched with any EUs, and can be initialized as  $\mathcal{S_{\rm UNMATCH}}=\{{\rm CU}_1,{\rm CU}_2,\cdots, {\rm CU}_N\}$.
We also define the set $\mathcal{S_{\rm MATCH}}$ to record the matched CUs and we initialize $\mathcal{S_{\rm MATCH}}$ with an empty set before matching process. After the initialization of the matching process, the initial matching operations will conduct to give each CU an EU to form a pair.
The main idea of this process is that each CU can propose to the EUs according to the preference lists.
The CU will first propose to its most preferred EU. If the EU has not been matched with any CU, these users will be paired together. Otherwise, the CU find another unmatched EU that can provide the minimum transmit power. Finally, all CUs and EUs are paired.
Then to further minimize the total transmit power, a swap matching process will be performed. In particular, the BS will iteratively check whether there are swap blocking pairs in the current matching.
 The matching process will terminate when there is no swap-blocking pair in the current matching.

\subsection{Performance Analysis of Algorithm \ref{alg:matching}}

In this subsection, we discuss the stability, convergence and computational complexity of  proposed swap matching based user paring algorithm.
\subsubsection{Stability and Convergence}
On the stability and convergence of Algorithm \ref{alg:matching}, we have the following proposition.
\begin{prop}
	{\bf (a)} If algorithm \ref{alg:matching} converges to a matching $\mathcal{ M}^{\star}$, then $\mathcal{ M}^{\star}$ is a two-side-stable matching.
	
	{\bf (b)} Algorithm \ref{alg:matching} converges to a two-side-stable matching after a limited number of swap operations.
	\end{prop}

\begin{proof}
	{\bf (a)}:  As stated in Algorithm \ref{alg:matching}, if it reaches to a ultimate matching $\mathcal{M}^*$, then there is no user pair can improve the performance and meanwhile doesn't hurt the performance of the other user pairs by exchanging matching. This implies that the current matching is the best matching according to Algorithm \ref{alg:matching} and there is no blocking pair. Thus the ultimate matching $\mathcal{ M}^*$ is a two-side-stable matching.
	
	{\bf (b)}: The convergence of Algorithm \ref{alg:matching} is determined by the swap-matching phase. Assume that the matching process changes with the swap operations as follow:
	\begin{align}
	\mathcal{ M}_0 \rightarrow \mathcal{ M}_1 \rightarrow \cdots \mathcal{ M}_\ell \rightarrow \cdots.
	\end{align}
	After swap matching $\ell$, the matching changes from $\mathcal{ M}_{\ell-1} \rightarrow \mathcal{ M}_\ell$.  According to \emph{Definition} 2, after each swap operation, at least one of the two related user-pairs get a decrease of the transmit power and the power consumption of the other user-pairs keep the same. Therefore, the total transmit power decreases after each swap operation. Notice
	that the total potential swap-blocking is finite since the number of users is limited, and also the total average power is lower-bounded due to the outage requirements of users. Thus there exists a swap operation $\ell^*$, after which the matching process will terminate and the total average power consumption stops decreasing. So Algorithm \ref{alg:matching} can be guaranteed to converge after several iterations. This completes the proof.
	\hfill $\blacksquare$	
	\end{proof}

\subsubsection{Computational Complexity}
The computational complexity of algorithm \ref{alg:matching} is consisted of two parts. One is from the  initial-matching phase to give an initial matching to all users, which has a complexity order of $\mathcal{O}(K)$. The other is due to the swap-matching operations in the swap-matching phase. For each ${\rm CU}_i$, there exist $K-1$ possible ${\rm EU}$s to do swapping, thus the complexity order is given by $\mathcal{O}(K(K-1))$.  Therefore, the total complexity is $\mathcal{O}(K^2)$. 
Compared to the optimal strategy using exhaustive search, which has a complexity order of $\mathcal{O}(K!)$, the computational complexity of the proposed swap-matching based algorithm is dramatically decreased.

\subsection{Solving Problem \eqref{p:power_min_ARQ_mu}}
In the above subsections, a swap-matching based algorithm is proposed to solve the user paring problem. As a result, the $2K$ users are grouped into $K$ user pairs, and thus problem \eqref{p:power_min_ARQ_mu} can be decoupled into $K$ subproblems. For each subproblem, it reduces to the power allocation problem with given user pairing. So Algorithm \ref{alg:sca} can be adopted to solve the subproblems in a parallel manner and thus problem \eqref{p:power_min_ARQ_mu} is solved.
In the following section, we will evaluate the performance of the proposed algorithms with numerical simulations.

\begin{figure}[!t]
	\centering
	\includegraphics[width=0.68\linewidth]{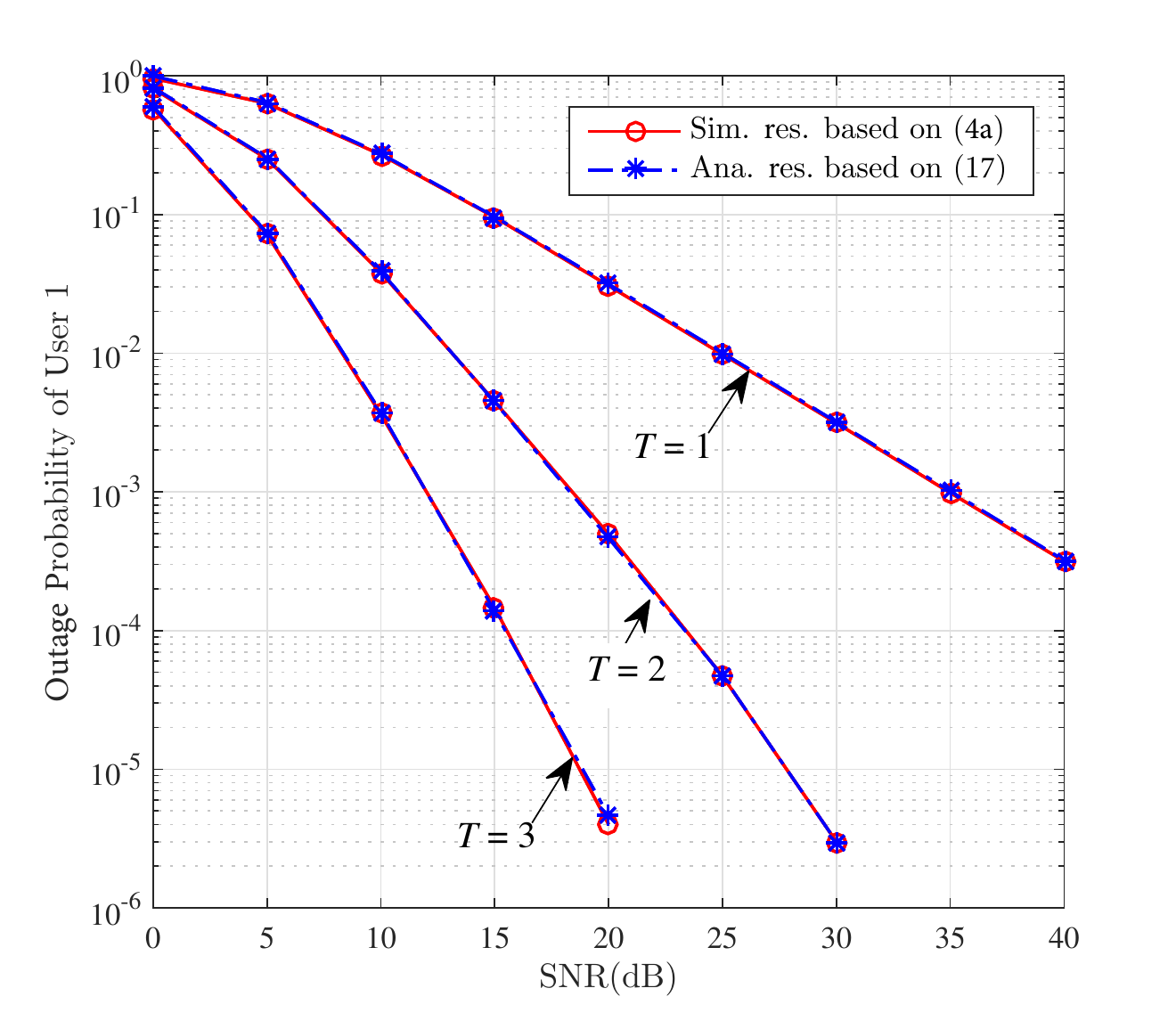}\\
	\caption{Outage probabilities of user 1 with $\frac{p_1}{p_2} = \frac{3}{2}$ and $d_1 = 10$ m.} \label{fig:outage_user1}
\end{figure}

\begin{figure}[!t]
	\centering
	\includegraphics[width=0.68\linewidth]{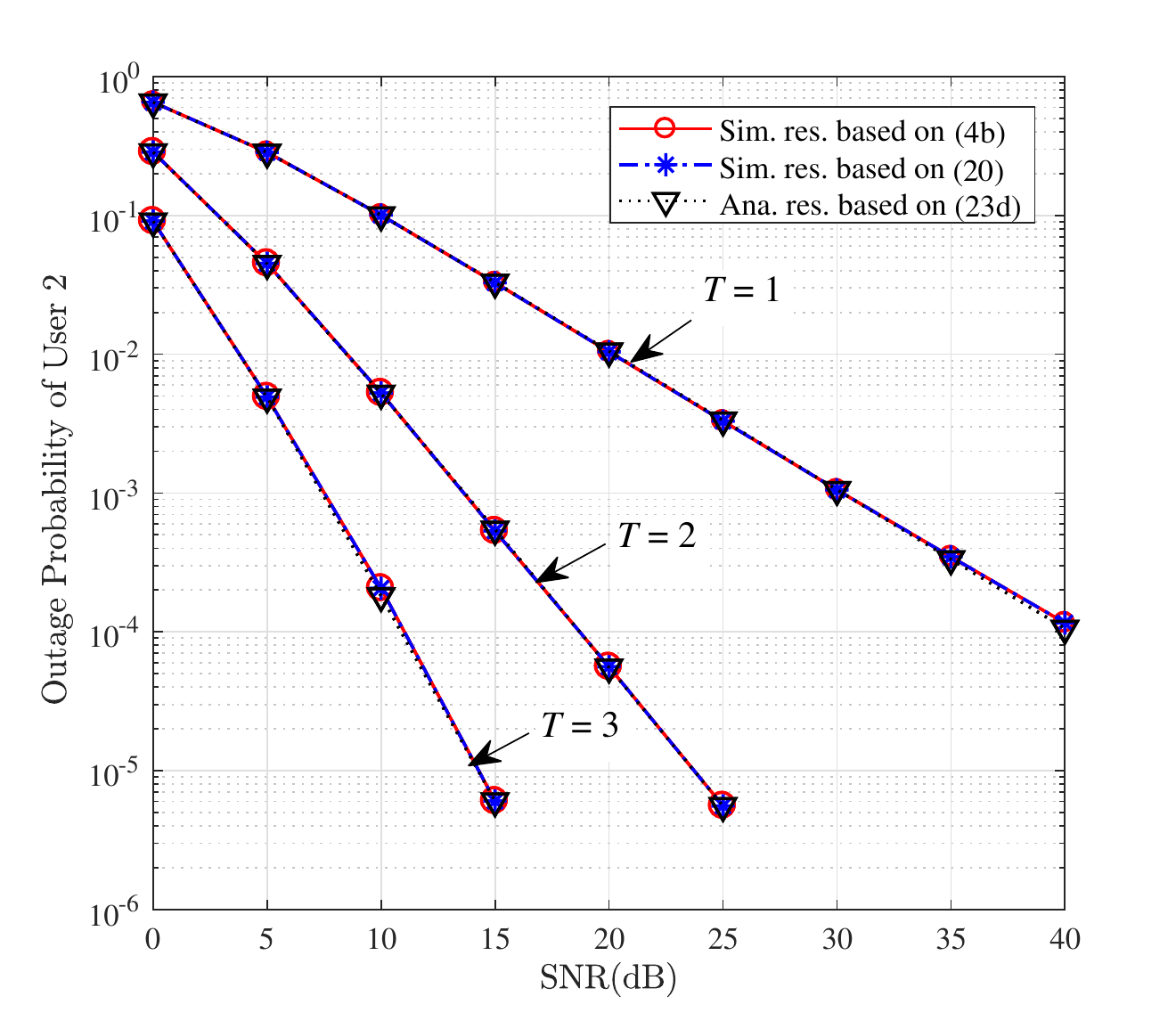}\\
	\caption{Outage probabilities of user 2 with $\frac{p_1}{p_2} = \frac{3}{2}$, $d_1 = 10$ m and $d_2 = 4$ m.} \label{fig:outage_user2}
\end{figure}

\section{Simulation Results} \label{sec:simulation}
	In this section, simulations are conducted to examine the efficacy of the proposed HARQ-CC enabled NOMA transmission scheme and algorithms.  Without loss of generality, in the simulations, we assume the users have the same outage probability requirements, i.e., $\delta_i = \delta_j= \delta,~ \forall i \in \kset_1, \forall j \in \kset_2$, the target SNRs of users are given by $\gamma_i = 0.2$ and $\gamma_j = 1,~\forall i \in \kset_1, \forall j \in \kset_2 $ \cite{cui-spl-2016}.
	The users are assumed to be randomly deployed in a circle region with $R_c = 4$ m and $R_e = 10$ m, and the path loss exponent is set to be $\alpha = 2$.
	The noise power is set to be $\sigma_1^2 = \sigma_2^2 = 0.1$.
	For the initialization in Algorithm \ref{alg:sca}, we simply set $p_{1,t}^{(0)} = 0.7 P_{\max}$ and $p_{2,t}^{(0)} = 0.3 P_{\max},~\forall t$. Then a feasible point of $\left\{x_{m,t}^{(0)},y_t^{(0)},{z}_t^{(0)}\right\}$ can be obtained based on \eqref{eq:CoV}.
	
\subsection{Accuracy of the Derived Outage Probability Expressions}	

We compare the analytical and simulation results of the outage performance 
in Fig. \ref{fig:outage_user1}, where the solid curves are simulation results based on \eqref{eq:outage_1}, and the dashed curves are the analytical results based on \eqref{eq:ana_outage_user1}. It can be seen that \eqref{eq:outage_1} can be tightly approximated by \eqref{eq:ana_outage_user1} even in the low SNR region and small number of retransmissions.
The outage performance of the analytical results of user 2 is given in Fig. \ref{fig:outage_user2},
where the solid and dashed curves are simulation results based on \eqref{eq:outage_2} and \eqref{eq:outage_approx_eq1} respectively, and the dotted curves are the analytical results based on \eqref{eq:outage_probability_user2}. One can observe that the approximation in \eqref{eq:outage_approx_eq1} and \eqref{eq:outage_probability_user2} are sufficiently tight in the HARQ-CC regime both in low and high SNR regions. It can be seen that  the outage probability characterized by \eqref{eq:outage_approx_eq1} is almost the same as that by \eqref{eq:outage_2} even in the low SNR region. Moreover the outage probability in \eqref{eq:outage_approx_eq1} can be conservatively approximated by the analysis expression in \eqref{eq:outage_probability_user2}.

\subsection{Performance Evaluation for Two-User Case}
To evaluate the performance of the proposed HARQ-CC enabled NOMA transmission scheme and the proposed adaptive power allocation (APA) scheme in Algorithm \ref{alg:sca}, three different schemes are considered as benchmarks and their rationales are described as follows:
\begin{itemize}
	\item \emph{NOMA with Exhaustive Search:} it yields the optimal performance by exhaustively search over all possible power allocation combinations of problem \eqref{p:power_min_ARQ} with the derived outage probabilities of \eqref{eq:ana_outage_user1} and \eqref{eq:outage_probability_user2}. The accuracy of the solutions can be adjusted when searching the transmit power and should be coincide with the other transmission schemes.
	\item \emph{NOMA with Equal Power Allocation (EPA):} the two users are scheduled by using NOMA and the transmit power in different retransmission rounds are equal.
	\item \emph{OMA with FDMA:} the whole frequent bandwidth are divided into two equal parts and allocated to the two users individually.
\end{itemize}

	\begin{figure}[!tp]\centering
	\includegraphics[width=0.68\linewidth]{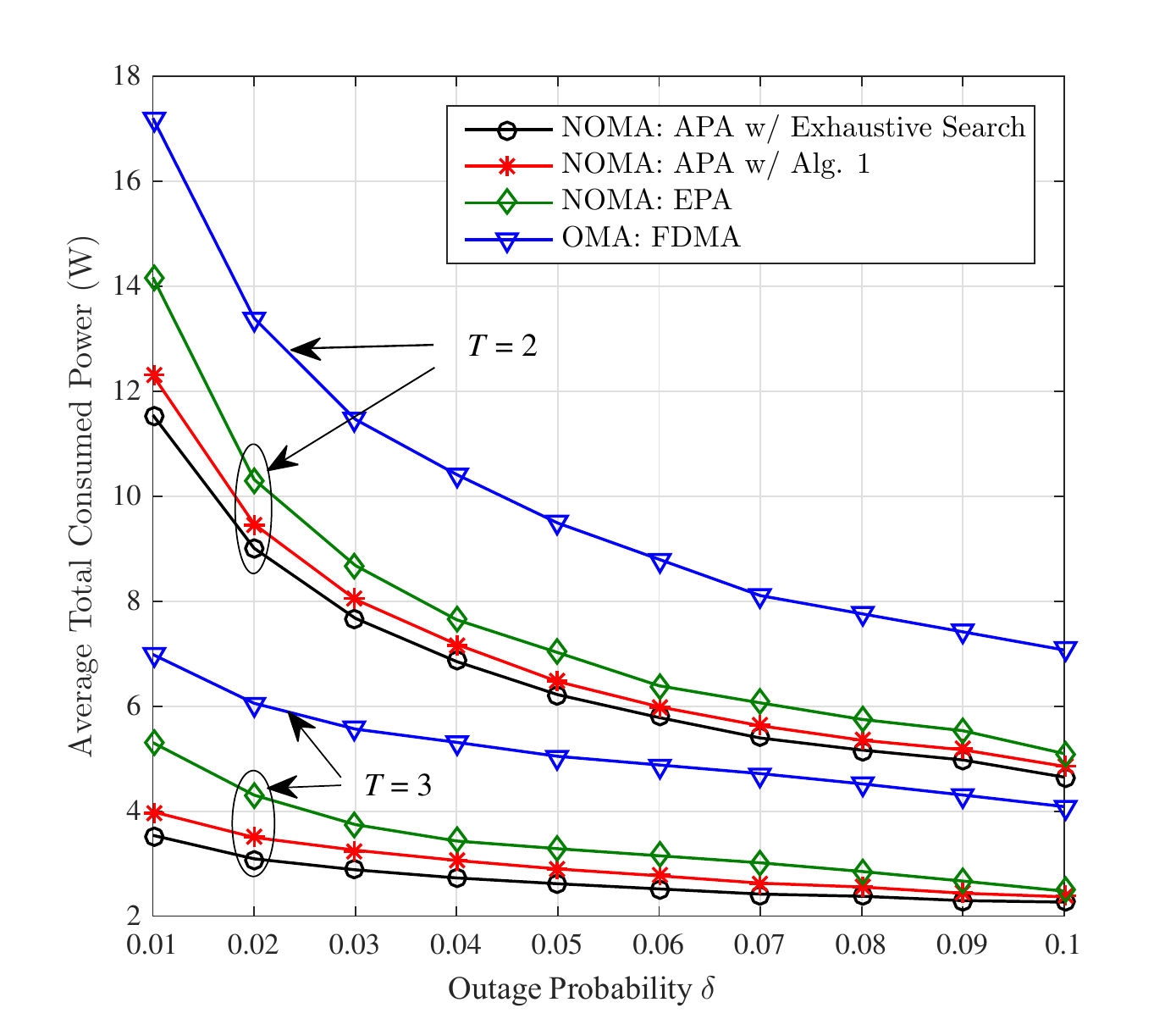}
	\caption{The total average power versus different outage requirements with $P_{\max} = 40$ Watt.} \label{fig:power_outage}
    \end{figure}

	To verify the relationship between the total average transmit power and the outage requirements of users, Fig. \ref{fig:power_outage} simulates the total average transmit power with different $\delta$'s under different maximum transmission rounds. It can be seen that the NOMA scheme can outperform the EPA and FDMA schemes under different $\delta$ and $T$.
	One can also observe that the performance of the proposed approximation approach and the successive approximation based algorithm is closed to the optimal performance attained by the exhaustive search approach.
	Moreover, it can also be seen that the average transmit power decreases with the increase of the tolerable outage probability (less strict outage requirements) of users, which indicates that retransmission would bring more performance improvement when users have more strict outage requirements.
	Furthermore, the power consumption gap between the equal power allocation scheme and the optimal one decreases with an increase of the tolerable outage probability, implying that it is more efficient to use more sophisticated when users have more strict outage requirements.
	
	
	\begin{figure}[!tp]\centering
		\includegraphics[width=0.68\linewidth]{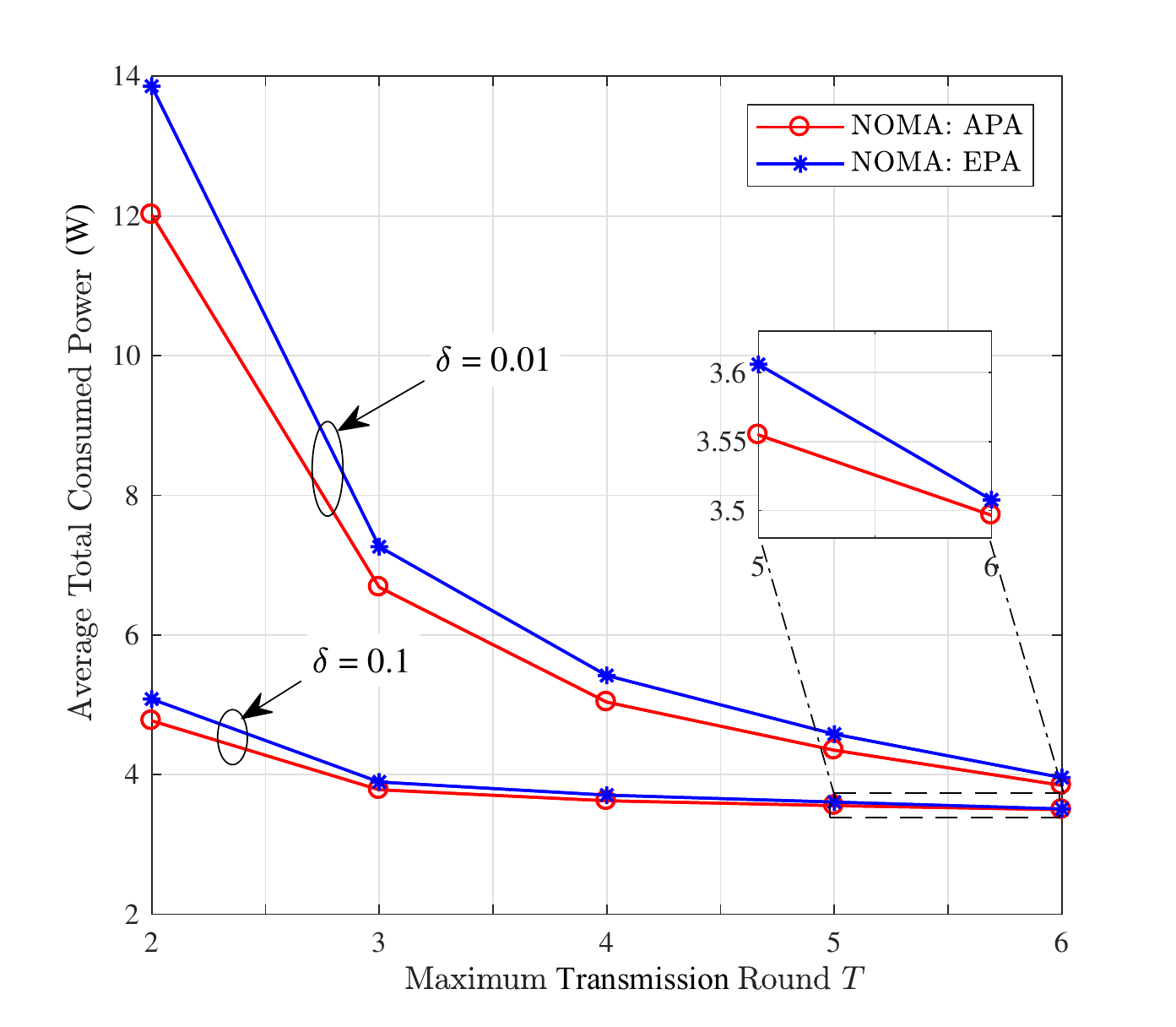}
		\caption{The total average power versus different outage requirements with $P_{\max} = 40$ Watt.} \label{fig:power_T}
	\end{figure}
In Fig. \ref{fig:power_T}, the relationship between the total average transmit power and the maximum retransmission round is investigated, and the EPA scheme is used as the benchmark. From Fig. \ref{fig:power_outage}, one can observe that the transmit power decreases with the increase of $T$, and the gap between APA and EPA schemes decreases with the increase of $T$. We can also observe that when the users have strict outage requirements e.g., $\delta = 0.01$, the total transmit power decreases quickly as increasing $T$.

\subsection{Performance Evaluation for Multi-User Case}

\begin{figure}[!tp]\centering
	\includegraphics[width=0.68\linewidth]{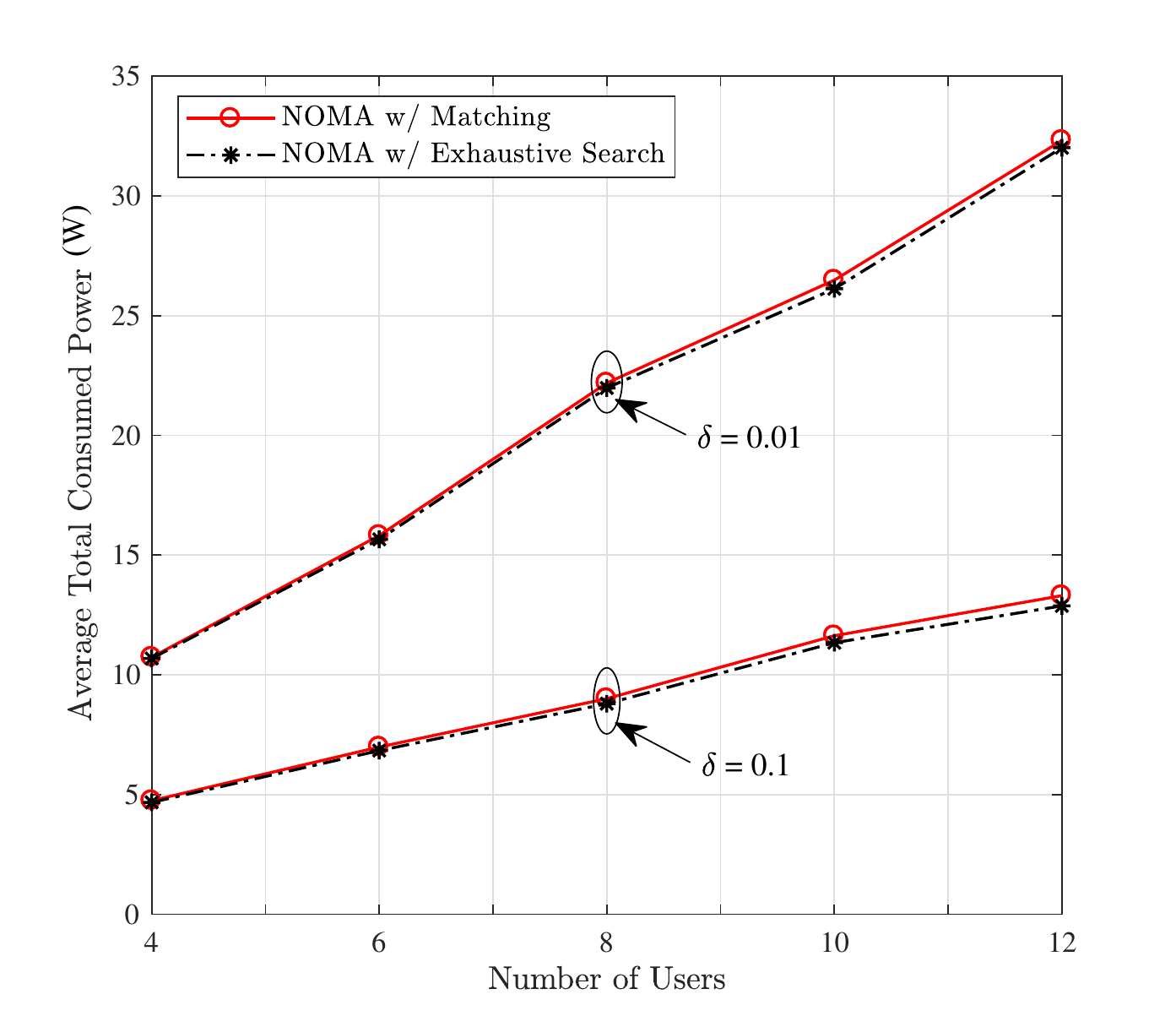}
	\caption{The total average power versus number of users in the system with $T = 3$ and $P_{\max} = 40$ Watt.} \label{fig:power_user}
\end{figure}
In this subsection, the performance of the proposed matching algorithm is examined, and the scheme by searching all the possible user-paring combinations will be served as the benchmark. In Fig. \ref{fig:power_user}, the performance of the matching algorithm is simulated with different number of users in the system by averaging $100$ user location realizations. From Fig. \ref{fig:power_user}, we can observe that when the numbers of users is small, e.g., $4$, the performance of the matching based algorithm is very close to the optimal performance by the exhaustive search scheme for different outage requirements; otherwise, the the matching algorithm will get a performance loss. But it is insignificant, as for example, when $K = 12$, the performance loss is only $0.9\%$ for $\delta = 0.01$ and $2.6\%$ for $\delta = 0.1$.

\begin{figure}[!tp]\centering
	\includegraphics[width=0.68\linewidth]{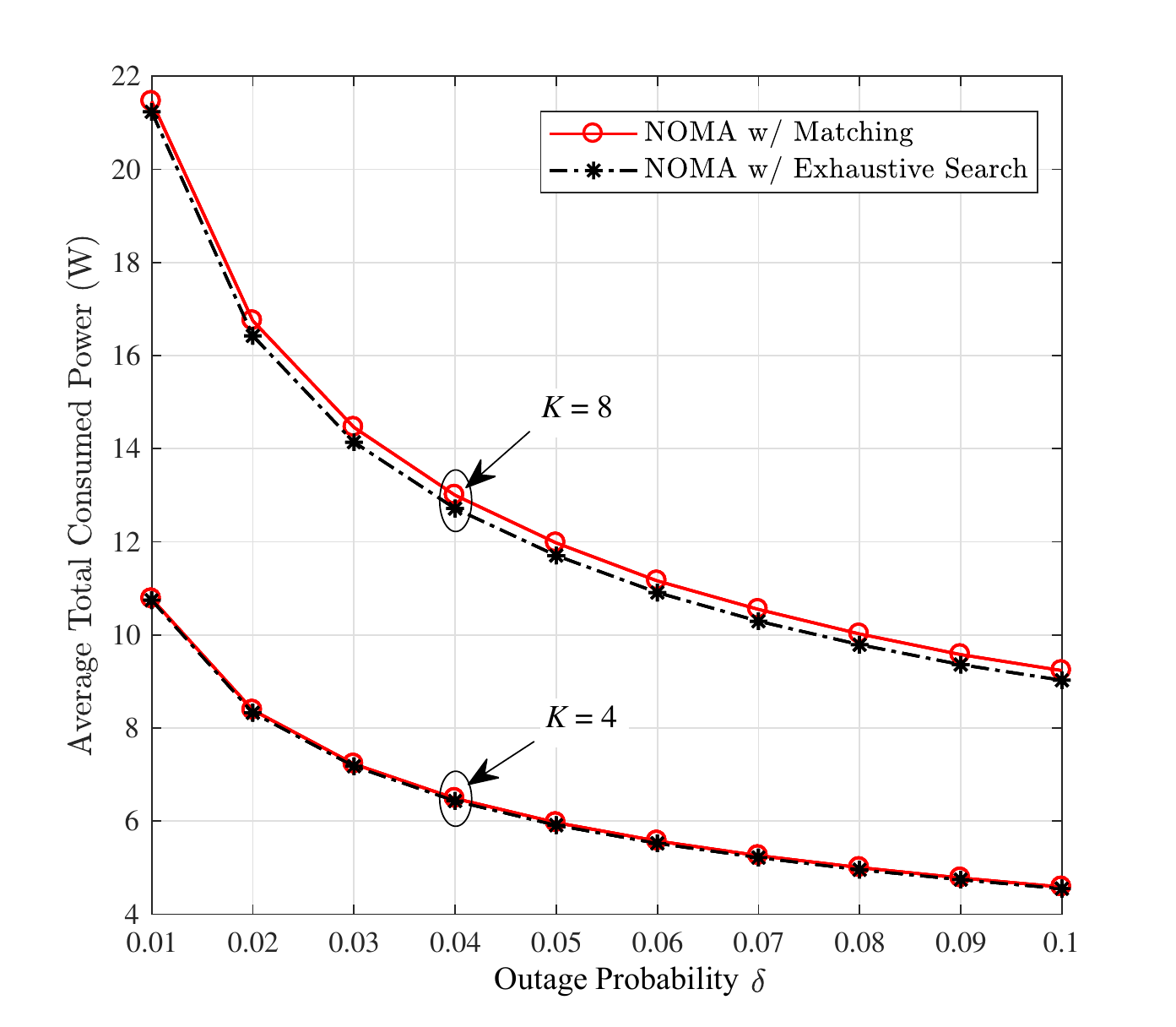}
	\caption{Total average power consumption versus outage requirement of users with $T = 3$ and $P_{\max} = 40$ Watt.} \label{fig:power_outage_mu}
\end{figure}
This can be seen in Fig. \ref{fig:power_outage_mu} more clearly where the performance of the proposed matching scheme under different outage requirements is also investigated. From Fig. \ref{fig:power_outage_mu}, we can see that for different outage requirements the matching scheme can always achieve the near-optimality performance, comparing with the optimal exhaustive search based algorithm.

\section{Conclusions}\label{sec:conclusion}
In this paper, we have considered the adaptive power allocation problems in HARQ-CC enabled NOMA systems subject to the outage probability constraints of users.
In particular, we have carefully derived the outage probabilities of the users to transform the original problem into an explicit one.
Furthermore, by properly approximating the formulated problem, a successive approximation based algorithm was proposed to iteratively solve the problem.
As a more practical scenario, the multi-user scenario was also investigated. The corresponding total average power minimization problem is boiled down to the joint design of user-paring and power allocation problem.
With the aid of matching theory, a suboptimal but with high-quality solution and low computational complexity algorithm has been presented to first solve the user-paring problem. Then it was shown that the remaining power allocation problem could be solved by the successive approximation based algorithm.
The simulation results have shown that the proposed approximation approach and the successive approximation based algorithm could achieve the near-optimal performance.
In addition, the proposed adaptive power allocation scheme can significantly outperform the equal power allocation scheme when users have strictly outage requirements.
Also, the matching based algorithm has been shown to have good accuracy compared with the optimal exhaustive searching method.

{
	\smaller[1]

}

\end{document}